\documentclass[lettersize,journal]{IEEEtran}
\usepackage{amsmath,amsfonts,amssymb}
\usepackage{array}
\usepackage{textcomp}
\usepackage{stfloats}
\usepackage{url}
\usepackage{verbatim}
\usepackage{graphicx}
\usepackage{cite}
\usepackage[colorlinks=true,linkcolor=blue]{hyperref}
\usepackage{amsthm}
\DeclareMathOperator*{\argmin}{argmin}
\usepackage{enumitem}
\usepackage{booktabs}
\usepackage{multirow}
\usepackage{xcolor}
\usepackage{arydshln}

\newtheorem{definition}{Definition}
\newtheorem{theorem}{Theorem}
\newtheorem{lemma}{Lemma}
\newtheorem{corollary}{Corollary}

\newtheorem{statement}{Statement}
\newtheorem{proposition}{Proposition}

\usepackage{subfigure}
\usepackage{placeins}
\usepackage{pifont}

\usepackage[linesnumbered, ruled, noend]{algorithm2e}  
\SetKwFunction{MyFunction}{MyFunction}
\SetKwProg{Fn}{Function}{:}{}

\newcommand{\SF}{SF}
\newcommand{\LL}{LL}
\newcommand{\CN}{CN}
\newcommand{\NL}{NL}
\newcommand{\OURALGO}{ABS}

\begin{document}

\title{
A Fragmentation-Aware Adaptive Bilevel Search Framework for Service Mapping in Computing Power Networks
}
\author{Jingzhao Xie, Zhenglian Li, Gang~Sun,~\IEEEmembership{Senior~Member,~IEEE}, Long Luo,~\IEEEmembership{Member,~IEEE}, and Hongfang~Yu,~\IEEEmembership{Senior~Member,~IEEE}
\thanks{Jingzhao Xie, Zhenglian Li, Gang~Sun, Long Luo, and Hongfang Yu are with University of Electronic Science and Technology of China, Chengdu 611731, China
(e-mail: jingzhaoxie@gmail.com; zhenglianli@foxmail.com; \{gangsun, llong, yuhf\}@uestc.edu.cn).
}
\thanks{This work was supported in part by National Natural Science Foundation of China (62394324).}
\thanks{Gang Sun is the corresponding author of this paper.
}
}


\markboth{Journal of \LaTeX\ Class Files,~Vol.~14, No.~8, August~2021}%
{Shell \MakeLowercase{\textit{et al.}}: A Sample Article Using IEEEtran.cls for IEEE Journals}



\vspace{-10em}
\maketitle
\begin{abstract}
Computing Power Network (CPN) unifies wide-area computing resources through coordinated network control, while cloud-native abstractions enable flexible resource orchestration and on-demand service provisioning atop the elastic infrastructure CPN provides.
However, current approaches fall short of fully integrating computing resources via network-enabled coordination as envisioned by CPN.
In particular, optimally mapping services to an underlying infrastructure to maximize resource efficiency and service satisfaction remains challenging.
To overcome this challenge, we formally define the service mapping problem in CPN, establish its theoretical intractability, and identify key challenges in practical optimization.
We propose Adaptive Bilevel Search (ABS), a modular framework featuring (1) graph partitioning-based reformulation to capture variable coupling, (2) a bilevel optimization architecture for efficient global exploration with best-response solving of local subproblems, and (3) fragmentation-aware evaluation for long-term performance guidance. 
Implemented using distributed particle swarm optimization, ABS is extensively evaluated across diverse CPN scenarios, consistently outperforming existing approaches. 
Notably, in complex scenarios, ABS achieves up to 73.2\% higher computing resource utilization and a 60.2\% higher service acceptance ratio compared to the best-performing baseline.
\end{abstract}

\section{Introduction}
\label{secs:intro}

\IEEEPARstart{T}{he} rapid development of emerging applications---including large-scale AI model training and inference, extended reality (XR) applications, and widespread Internet of Things (IoT) ecosystems---is making network services increasingly complex~\cite{com2net}.
These novel services exhibit diverse requirements and dynamic characteristics, consume substantial computing resources, and necessitate flexible, on-demand provisioning across all scenarios~\cite{standing}.
However, in this context, conventional network architecture with dedicated infrastructure struggles to satisfy diverse service requirements and often creates resource silos, leading to resource waste and service dissatisfaction~\cite {CNC}.
In response, Computing Power Network (CPN) has emerged, integrating distributed heterogeneous computing resources via the network to form an infrastructure characterized by resource pooling and service elasticity, thereby enabling on-demand resource provisioning across domains~\cite{CNC, det3}.

\begin{figure}[!tp]
    \centering
    \rotatebox{0}{\includegraphics[width=0.8\linewidth]{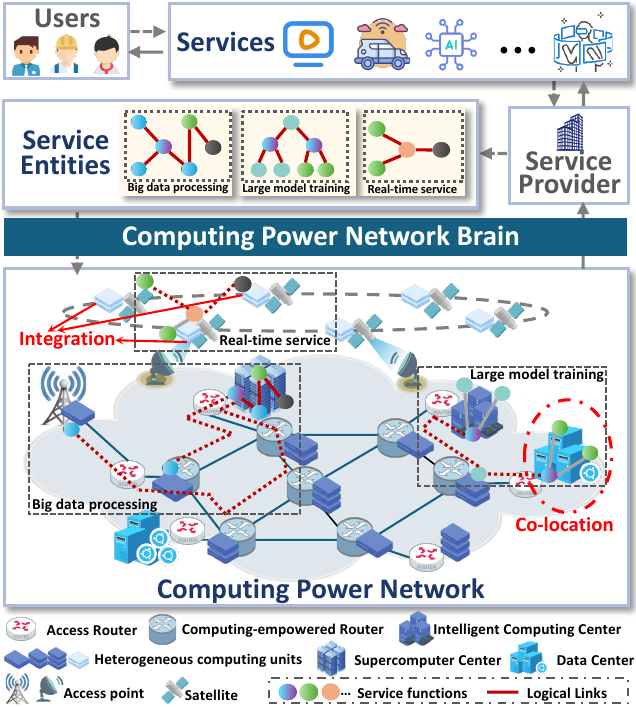}}
    \caption{
    Conceptual overview of service provisioning and mapping in the CPN environment.}
    \label{fig: architecture}
    \vspace{-10pt}
\end{figure}

To fully realize the potential of CPN, adopting a cloud-native perspective---characterized by declarative service provisioning and on-demand resource supply---is widely considered essential~\cite{service-oriented, CNC}.
The cornerstone of this approach is decoupling services from physical resources and abstracting them into self-contained logical entities known as service entities (SEs). 
SEs define the essence and behavior of services, typically comprising service functions (SFs) that provide business capabilities, logical links (LLs) transmitting data flows between SFs, and descriptions of business logic and resource requirements~\cite{service-oriented, serviceentities}.
Fig.~\ref{fig: architecture} illustrates this model within CPN: users specify desired service capabilities, which the service provider translates into SEs.
These SEs are then dynamically mapped onto the underlying physical infrastructure by the CPN brain~\cite{CNC}, which orchestrates resources in accordance with real-time infrastructure conditions.
This mapping process enables flexible service provisioning and dynamic resource allocation. 
More importantly, it lays the foundation for integrating distributed computing resources---the central objective of CPN.
Specifically, service functions of SEs are mapped to distributed computing nodes and interconnected through logical links, enabling previously isolated resources to collaborate over physical networks as a cohesive system that delivers unified services~\cite{serviceoriented}.
Consequently, the network evolves from a mere connectivity medium into a critical enabler for this integration, which in turn requires coordinated control of both \SF\ placement and \LL\ mapping.

Existing mapping approaches can be broadly categorized into Directed Acyclic Graph (DAG) models~\cite{dagsurvey, dagqoe, dagdrl}, Service Function Chaining (SFC)~\cite{sfccolocate, monte}, and Virtual Network Embedding (VNE)~\cite{originvne, rmd}.
Among them, VNE enables fine-grained resource allocation and supports arbitrary SE topologies, making it a natural starting point for CPN.
In comparison, DAG-based methods typically adopt coarse-grained network abstractions, while SFC-based ones are primarily tailored to chain-like service structures.
Beyond representational expressiveness, another key consideration for CPN is \SF\ co-location: modern platforms make it an important optimization knob (e.g., via containerization and virtual links~\cite{klonet}), as strategically co-locating \SF s can improve node utilization~\cite{beware} and reduce network overhead by shifting traffic to intra-node or intra-domain communication~\cite{megate}.
Such benefits directly align with the CPN objective of network-enabled computing integration.
However, while many existing formulations permit co-location~\cite{25vne}, VNE-style pipelines typically treat it as a feasibility option rather than an optimization lever, i.e., they do not explicitly decide which \SF s to co-locate or account for how co-location reshapes resource usage and overall system performance~\cite{energyvne, pgcnn, aivne, fan2023node, tccn5g}.
Moreover, theoretical guidance and systematic analysis remain lacking on how \SF s should be co-located to fully exploit co-location opportunities under online arrivals and long-term objectives.
Specifically, explicitly optimizing co-location introduces three key challenges: (i) it tightly couples computing and network mapping decisions, (ii) it expands the decision space exponentially, making search prone to poor local optima, and (iii) it renders myopic per-request metrics insufficient to capture long-term fragmentation and future admissibility (details in Section~\ref{secs:model}).

In this work, we propose a Service Entity Mapping (SEM) scheme that extends typical VNE by explicitly optimizing \SF\ co-location as an integral part of the decision process.
We focus on developing a general framework for solving SEM problems, aiming to fully leverage the network's enabling capabilities to enhance the utilization of computing resources and improve service satisfaction.
The main contributions are summarized as follows:

\begin{enumerate}
   \item We introduce Service Entity Mapping (SEM), an extension of VNE tailored to CPN environments. SEM elevates \SF\ co-location from a feasibility option to an explicit mapping decision, capturing compute–network coupling and enabling more effective integration of distributed computing resources.
   \item We formulate the online SEM problem as a mathematical model that optimizes long-term system performance by balancing service satisfaction and computing resource utilization against maintaining network overhead. We formally prove the problem’s NP-hardness, and discuss key challenges for practical optimization.
   \item We develop an Adaptive Bilevel Search (ABS) framework for solving the SEM problem. This framework integrates a graph partitioning perspective for generating co-location-aware mapping decisions, a bilevel optimization structure combined with an adaptive feedback mechanism for robust global exploration, and resource fragmentation metrics for long-term performance evaluation. This comprehensive design empowers ABS to iteratively refine mapping decisions via feedback, bridging global and local optimization and enabling continuous evolution toward overall system optimality. 
   \item We validate the efficacy and superiority of the proposed ABS framework through extensive simulations. Specifically, we implement ABS by leveraging a distributed particle swarm optimization algorithm. Our comprehensive experimental results demonstrate that ABS significantly outperforms state-of-the-art methods across key performance metrics, including system profit, service request acceptance ratio, and long-term average revenue. 
\end{enumerate}

The rest of the paper is organized as follows.
Section~\ref{secs:related} reviews related works.
Section~\ref{secs:model} formulates the mathematical model of the online SEM problem, analyzes its theoretical and practical difficulties, and provides key solution insights.
Section~\ref{sec: framework} proposes the design of the ABS framework and presents its implementation details.
Section~\ref{secs:evaluation} discusses the simulation results.
Finally, Section~\ref{secs:conclusion} concludes the paper.

\section{Related Work}
\label{secs:related}

We present the service mapping schemes and discuss key technologies for \SF\ co-location in CPN environments.

\subsection{Service Mapping Schemes}

\textbf{DAG-based mapping:}
It models service entities as DAGs, where nodes represent service functions and edges capture their dependencies~\cite{dagsurvey}. 
This structure clarifies functional relationships and facilitates efficient mapping of \SF s to distributed nodes.
For example, Fan et al.~\cite{dagqoe} studied task scheduling in distributed networks, proposing a heuristic hierarchical multi-queue algorithm to determine the execution order and mapping of \SF s jointly.
Goudarzi et al.~\cite{dagdrl} investigated service mapping in edge and fog computing by modeling SEs as DAGs. 
To improve adaptability, a distributed deep reinforcement learning (DRL) approach is employed, utilizing parallel actors and aggregated experiences for policy updates.
However, DAG-based methods primarily portray logical dependencies, often simplifying network connections as point-to-point channels that only reflect transmission delays~\cite{dagdrl}. 
This abstraction overlooks the allocation of intermediate network links, which limits fine-grained coordination between computing and network resources and thus constrains computational integration in network-constrained CPN scenarios.

\textbf{SFC-based mapping:}
Service entities are represented as linear function chains, explicitly considering fine-grained network links and corresponding resource requirements between service functions. 
He et al.\cite{serviceoriented} addressed heterogeneous resource orchestration in space-air-ground integrated networks by modeling SFC mapping and resource allocation as a nonconvex optimization problem and proposing an alternating iterative algorithm. 
Dong et al.\cite{standing} proposed a federated meta-reinforcement learning method for fast, privacy-preserving policy adaptation in SFC-based 6G network slicing, while Zhang et al.\cite{sfcmigration} introduced a heuristic for joint SFC placement and virtual network function migration in multi-domain networks.
Most of these works do not explore the co-location of \SF s. 
Notably, Wang et al.\cite{sfccolocate} investigated co-location strategies to improve resource utilization and latency performance, and proposed a parallelized SFC deployment algorithm on a heterogeneous architecture for enhanced flexibility and performance.
Nevertheless, the linear chaining model of SFC limits its implementation flexibility to support complex topologies and diverse service requirements, making it less suitable for CPN-oriented service mapping.

\textbf{VNE-based mapping:}
Compared to DAG and SFC, the VNE scheme is better suited for service mapping in CPN environments due to its support for both fine-grained network resource allocation and arbitrary SE topologies. 
Zhang et al.~\cite{pgcnn} proposed a DRL-based Quality-of-Service (QoS)-aware VNE algorithm for space-air-ground-ocean integrated networks, which utilizes K-means clustering to classify each request's QoS requirements and dynamically optimizes mapping strategies for improved resource utilization.
Geng et al.~\cite{gal} proposed GAL-VNE, which combines global reinforcement learning with local one-shot neural prediction. 
This approach globally optimizes resource allocation across requests and performs mapping efficiently at the local level.
Guan et al.~\cite{multidimensional} addressed multidimensional resource fragmentation to optimize edge resource utilization, and proposed a fragmentation-aware bilevel VNE algorithm that enhances request acceptance and infrastructure profit.
However, these approaches remain limited as they do not explicitly leverage \SF\ co-location to pack their demands more efficiently into node capacities or localize traffic to alleviate bandwidth bottlenecks, hindering full network-enabled computing integration in CPN scenarios.
Zhang et al.~\cite{rmd} explicitly allowed multiple virtual nodes from the same virtual network to be co-located on a single substrate node to improve mapping efficiency. 
However, as a heuristic algorithm tailored for network emulation
platforms, this work is not generally applicable to CPN
environments and is still prone to local optima.

\subsection{Service Function Co-location Technologies}

\textbf{\SF\ Instantiation:} 
Containerization has become a foundational technology in cloud-native environments due to its lightweight virtualization and efficient resource isolation capabilities. 
By encapsulating each service function within its own container, multiple \SF s can be flexibly co-located on a single host, thereby improving resource utilization and deployment agility~\cite{service-oriented}. 
Such co-location is enabled by Linux namespaces and control groups (or
cgroups)~\cite{klonet}, which provide strong isolation between containers and allow for the independent management of resources, such as CPU and memory.
In addition to CPU and memory resources, various technologies have emerged to enable efficient GPU sharing and management among containers~\cite{beware}.

\textbf{\SF\ Communication:}
At the node level, co-located \SF s can flexibly establish arbitrary Layer 2 communication networks via technologies such as veth pairs and Open vSwitch (OVS), enabling intra-host data transmission and reducing inter-node network consumption~\cite{klonet}. 
At the domain level, \SF s deployed in different data centers within a domain can establish communication links through tunneling technologies such as VxLAN, allowing traffic to remain within the domain and thus preserving valuable inter-domain network resources.
Furthermore, per-flow traffic engineering (TE) tunnels enable logical links to be precisely mapped onto arbitrary forwarding paths in the CPN network, allowing for flexible and fine-grained network control~\cite{megate}.

Nevertheless, there remains a lack of theoretical guidance on how to exploit these new capabilities for service mapping in CPN. 
Therefore, it is imperative to develop new theoretical models and solution frameworks that can systematically leverage advanced co-location mechanisms to optimize service satisfaction and resource utilization. 
\section{
Model Formulation and Problem Study}
\label{secs:model}
This section details the system model and mathematical formulation of the SEM problem. We then establish its theoretical intractability, analyze its practical challenges, and provide key insights for designing effective solutions.

\subsection{System Model}
\subsubsection{Computing Power Network (CPN)}
The underlying CPN topology is formally defined as an undirected graph $G^s = (N^s, L^s)$~\cite{monte}, where $N^s$ represents the set of computing nodes (CNs) and $L^s$ denotes the set of their interconnecting network links (NLs).
In this paper, the superscript $s$ corresponds to the CPN topology.
Each CN $m^s \in N^s$ is characterized by its available computing capacity, denoted as $C(m^s)$.
Similarly, each NL connecting a pair of CNs, denoted by $l_{mn}^s \in L^s$ ($\forall m^s,n^s \in N^s$), is associated with a bandwidth capacity $B(l_{mn}^s)$.

\subsubsection{Service Entity (SE)}
Each service request $i$ corresponds to an SE, modeled as an undirected graph $G^v_i = (N^v_i, L^v_i)$~\cite{microservicevne}.
In this representation, $N^v_i$ represents the set of service functions (SFs), \nolinebreak and $L^v_i$ denotes the set of logical links (LLs) connecting these \SF s.
For clarity of exposition, we represent the computing requirement of each \SF\ $u^v \in N^v_i$ by a scalar demand $c(u^v)$; each \LL\ $l^v_{uw}\in L_i^v$, connecting a pair of \SF s $u^v,w^v\in N_i^v$, is associated with a bandwidth requirement of $b(l^v_{uw})$.
This scalar abstraction can be extended to heterogeneous multi-dimensional node resources by replacing $c(u^v)$ with a demand vector and applying per-dimension capacity constraints.

\subsection{
Problem Description and Formulation
}
We address the SEM problem in an online setting where service requests arrive sequentially.  
The main challenge is to map each incoming SE onto the CPN in real time under current resource availability, where admitted SEs occupy computing and network resources for their lifetimes. 
Given a CPN topology $G^s = (N^s, L^s)$ and an SE request $G^v_i = (N^v_i, L^v_i)$, the SEM problem is abstracted into two phases: the service function mapping (SFnM) and the logical link mapping (LLnM).

\subsubsection{Service Function Mapping (SFnM)}
This phase maps each SF of $G^v_i$ to an appropriate CN in $G^s$. A single CN can host multiple SFs (SF co-location), and each mapped SF consumes specified computing resources. We introduce a binary variable $x^{u^v}_{m^s} \in \{0, 1\}$ to indicate if SF $u^v \in N^v_i$ is mapped to a CN $m^s\in N^s$. The SFnM phase is subject to the following constraints:

$\bullet$ \emph{SF-to-CN mapping:} Each SF must be mapped to exactly one CN (Equation~\eqref{f1}).
\begin{equation}
\forall u^v \in N_i^v:\quad \sum_{m^s \in N^s}\, x^{u^v}_{m^s} = 1.
	\label{f1}
\end{equation}

$\bullet$ \emph{CN-to-SF hosting:} A CN can host zero or more SFs (Inequality~\eqref{f2}).
\begin{equation}
	\forall m^s \in N^s:\quad \sum_{u^v \in N_i^v}\, x^{u^v}_{m^s} \geq  0.
	\label{f2}
\end{equation}

$\bullet$ \emph{CN resource capacity:} The total computing requirements of hosted SFs must not exceed the CN’s available capacity (Inequality~\eqref{f3}). 
\begin{equation}
	\forall m^s \in N^s: \quad
\sum_{u^v \in N_i^v} x_{m^s}^{u^v} \cdot c(u^v) \leq C(m^s).
\label{f3}
\end{equation}

\subsubsection{Logical Link Mapping (LLnM)}
The SF-to-CN mapping is fixed in this phase: an \LL\ whose endpoint SFs are mapped to the same CN is internally hosted, whereas a \textit{Cut}-\LL\ (endpoints mapped to different CNs) must be routed through the network~\cite{megate}.
Although $G^s$ is undirected, we use a canonical ordering of CN pairs for indexing and define $K \triangleq \{(m^s,n^s)\in N^s\times N^s : m^s < n^s\}$.
For each $k\in K$, we pre-compute a set of loop-free tunnels $P^k$, and let $P^s=\bigcup_{k\in K}P^k$.
Parameter $y^p_{mn}\in\{0,1\}$ indicates whether tunnel $p\in P^s$ traverses \NL\ $l^s_{mn}$.
Since $G^v_i$ is undirected, for each \textit{Cut}-\LL\ with endpoints $\{u^v, w^v\}$, we let $(u^v, w^v)$ be the ordered pair such that $u^v$ is mapped to $m^s$ and $w^v$ is mapped to $n^s$ with $m^s < n^s$.
We then refer to this unique pair as $k=(m^s, n^s) \in K$.
For internally hosted \LL s, all $f$-variables are set to zero.
We impose the following constraints: 

$\bullet$ \emph{Flow indivisibility:} 
Each \textit{Cut}-\LL\ is mapped to at most one tunnel (Inequality~\eqref{f4}); internally hosted \LL s select no tunnel.
\vspace{-0.5em}
\begin{equation}
	\forall l^v_{uw} \in L_i^v:\quad 
\sum_{k \in K}
\sum_{p \in P^k} f^{u^vw^v}_{p,k} 
 \leq 1 .
	\label{f4}
\end{equation}

$\bullet$ \emph{Tunnel selection coherence:} A tunnel can only be selected for an LL if its endpoint SFs are mapped to the tunnel’s respective CNs (Inequality~\eqref{f5}).
\begin{align}
\forall l^v_{uw}\in L^v_i,\ \forall k=(m^s,n^s)\in K,\ \forall p\in P^k:\nonumber\\
f^{u^vw^v}_{p,k} &\le x^{u^v}_{m^s}, \nonumber\\
f^{u^vw^v}_{p,k} &\le x^{w^v}_{n^s}.
\label{f5}
\end{align}

$\bullet$ \emph{Network Link (NL) resource capacity:} The total bandwidth requirements of mapped LLs must not exceed the capacity of any traversed NL (Inequality~\eqref{f6}).
\begin{align}
	&\forall l_{mn}^s \in L^s : \nonumber \\
\sum_{l_{uw}^v \in L_i^v} \sum_{k \in K} \sum_{p \in P^k} &  f^{u^vw^v}_{p,k} \cdot y^p_{mn} \cdot b(l^v_{uw})
\leq B(l^s_{mn}).
\label{f6}
\end{align}

\textbf{Problem Formulation.} 
Upon receiving a new service request, the CPN orchestrator solves the SEM problem to optimize  SFnM and LLnM decisions, aiming to maximize \emph{long-term system profitability}. This Profit metric balances service acceptance, revenue, and resource costs. The SEM problem is modeled as ($\mathbb{P}1$):
\begin{align}
	&(\mathbb{P}1): 
	\max_{\mathbf{x^t,f^t}} 
        \label{f15} \\ \nonumber
	 &\underbrace{(p_{ac}(t))^\varkappa}_{\substack{\text{Weighted}\\\text{Acceptance Ratio}}} \cdot \underbrace{\biggl( \underbrace{\sum_{G^v_i \in \Upsilon^a(t)} \mathcal{R}(G_i^v)}_{\text{Revenue}} - \underbrace{\omega \cdot \sum_{G^v_i \in \Upsilon^a(t)} \mathcal{C}(G_i^v)}_{\text{Weighted Cost}} \biggr)}_{\text{Net Revenue (Profitability Core)}}
	  \\ \nonumber
	&s.t. \quad  (\ref{f1}) - (\ref{f6}), \quad \quad \forall t.
\end{align}

Here, $\mathbf{x^t}$ and $\mathbf{f^t}$ represent the SF and LL mapping decisions at time $t$. The objective (\ref{f15}) balances the weighted service acceptance ratio with net revenue. 
The coefficient $\varkappa \geq 1$ emphasizes the importance of the acceptance ratio, directly enhancing service satisfaction. The factor 
$0< \omega < 1$ allows trading network resource consumption for higher computing resource utilization, aligning with CPN’s goal of maximizing distributed computing resource use.  
The core components of the objective function are: 

\textit{Acceptance ratio $p_{ac}(t)$}: The proportion of successfully mapped SEs up to time $t$.
\begin{equation}
	p_{ac}(t) = \frac{|\Upsilon^a(t)|}{|\Upsilon(t)|},
	\label{f9}
\end{equation}
where $\Upsilon(t)$ is the set of requested SEs and $\Upsilon^a(t)$ is the subset of successfully mapped ones.

\textit{Revenue $\mathcal{R}(G_i^v)$}: Revenue from an accepted $G_i^v$, derived from its SF and LL resource demands. Rejected requests yield zero revenue. 
\begin{align}
	\mathcal{R}(G_i^v) = w_c \sum_{u^v \in N^v_i} c(u^v) 
	+ w_b \sum_{l^v_{uw} \in L_i^v} b(l^v_{uw}).
	\label{f7}
\end{align}

\textit{Cost $\mathcal{C}(G_i^v)$}: Resources consumed to host $G_i^v$ comprising computing node cost $\mathcal{C}_n(G_i^v)$ and network link cost $\mathcal{C}_l(G_i^v)$. 
\begin{align}
&\mathcal{C}(G_i^v) =
\pi_c \, \mathcal{C}_n(G_i^v) + \pi_b \, \mathcal{C}_l(G_i^v) \nonumber \\
&= \pi_c\sum_{u^v \in N^v_i} c(u^v) \nonumber \\
&+ \pi_b\sum_{l^v_{uw} \in L_i^v} \sum_{k \in K} \sum_{p \in P^k} \sum_{l^s_{mn} \in L^s}
f^{u^vw^v}_{p,k} \cdot y^p_{mn} \cdot b(l^v_{uw}).
\label{f8}
\end{align}
Computing cost equals SF resource demands. 
Network cost varies based on the number of NLs used for LL mapping.

Here, $w_c$ and $w_b$ are the revenue weights for computing and bandwidth demands, respectively, and $\pi_c$ and $\pi_b$ balance the \SF\ and \LL\ cost contributions.

\subsection{Theoretical Analysis of SEM}
Since each admission decision affects the future system state, ($\mathbb{P}1$) couples decisions across time.
Solving such an online integer linear programming (ILP) problem to optimality in real time is generally computationally intractable.
For this reason, a widely adopted approach is to make independent decisions for each service request based on the current system state~\cite{divide,gastp}.
Thus, we adopt this time-local greedy strategy that simplifies ($\mathbb{P}1$) at the expense of global optimality.
For a newly arrived request at time $\tau$, we freeze the system state at $\tau$ and transform ($\mathbb{P}1$) into a per-request snapshot optimization problem ($\mathbb{P}2$), formulated as follows\footnote{For notational convenience, we use subscript $i$ to denote the service request arriving at time $\tau$.}:
\begin{subequations}
	\begin{align}
	&(\mathbb{P}2):
	\max_{\mathbf{x,f}}
	(p_{ac}(\tau))^{\varkappa} \cdot \Big(\mathcal{R}(G^v_i) - \omega \cdot \mathcal{C}(G^v_i, \mathbf{x,f})\Big)  \\
	&s.t. \quad  (\ref{f1}) - (\ref{f6}).
	\end{align}
	\label{f16}
\end{subequations}

($\mathbb{P}2$) is time-local and single-request: since $p_{ac}(\tau)$ is fixed at time $\tau$, the objective depends only on the current request and state, avoiding explicit reasoning about future arrivals and making ($\mathbb{P}2$) more amenable to practical algorithms than the full online problem.
However, this simplification does not make the problem tractable:  ($\mathbb{P}2$) remains NP-hard (see Theorem~\ref{T1}), and the full online formulation \ ($\mathbb{P}1$) is NP-hard as well (see Corollary~\ref{c1}).


\begin{statement}[($\mathbb{P}2$)'s equivalent version ($\mathbb{P}2\text{-a}$)]
	At time $\tau$, as implied by~(\ref{f9}), (\ref{f7}), and~(\ref{f8}), the acceptance ratio, request revenue, and SF mapping cost are constants. Consequently, the objective in \textup{($\mathbb{P}2$)} reduces to minimizing the bandwidth cost incurred in LLnM, leading to the following simplified yet equivalent formulation, which serves as the basis for proving the intractability of \textup{($\mathbb{P}2$)} and the more challenging \textup{($\mathbb{P}1$)}:
    \begin{subequations}
	\begin{align}
	&(\mathbb{P}2\textup{-a}):
	\min_{\mathbf{x,f}}
	\mathcal{C}_l(G^v_i, \mathbf{x,f})   \\
	&s.t. \quad  (\ref{f1}) - (\ref{f6}).
	\end{align}
	\label{f17}
\end{subequations}
\end{statement}

\begin{statement}[($\mathbb{P}2$)'s decision version, ($\mathbb{P}2$\text{-D})]
	Given the CPN topology $G^s = (N^s, L^s)$, the request's SE topology $G^v_i = (N^v_i, L^v_i)$ at time $\tau$, and a positive number $Q$, the problem is to ask whether there exists a mapping solution that keeps the LLnM cost (\ref{f17}a) no greater than $Q$ while satisfying the constraints (\ref{f17}b).
\end{statement}


\begin{statement}[Graph bisection problem (GBP)'s decision version, GBP-D]
	Given an undirected graph $G(V, E)$ and an integer $K$ within the range $0 < K \leq |E|$, the problem involves determining whether $V$ can be divided into two disjoint, \nolinebreak equally sized vertex subsets $V_1$ and $V_2$ (i.e., $||V_1|-|V_2|| \leq 1$), such that their union spans all vertices ($V_1 \cup V_2 = V$) and their intersection is empty ($V_1 \cap V_2 = \emptyset$), while ensuring the number of edges connecting $V_1$ and $V_2$, denoted as $|Cut(V_1, V_2)|$, does not exceed $K$.
\end{statement}

\begin{theorem}
	\textup{($\mathbb{P}2$)} is NP-hard. 
	\label{T1}
\end{theorem}
\begin{proof}
    To prove it, we reduce GBP-D to ($\mathbb{P}2$)-D.
Let $\{G(V, E), V_1, V_2, K\}$ be an arbitrary instance of GBP-D.
On this basis, we construct an instance of ($\mathbb{P}2$)-D with a CPN topology $G^s = (N^s, L^s)$, a request's SE topology $G^v_i = (N^v_i, L^v_i)$, and a threshold $Q > 0$.
Then, we prove that a feasible ($\mathbb{P}2$)-D solution exists if and only if there exists a valid graph bisection for $\{G(V, E), V_1, V_2, K\}$~\cite{novel}.

First, we design a polynomial-time transformation, in which we construct the SE topology $G^v_i = (N^v_i, L^v_i)$ and the CPN topology $G^s = (N^s, L^s)$ according to the given GBP-D instance.
For the SE topology, let $N^v_i = V$ and $L^v_i = E$, with unit resource requirements for each \SF \ and \LL.
The CPN topology comprises three \CN s: $N^s = \{n_1, n_2, n_3\}$ and two \NL s: \ $L^s = \{l^s_{n_1 n_2}, l^s_{n_2 n_3}\}$, with $n_1$ containing $|V_1|$ resource units, $n_2$ containing $|V_2|$ resource units, and $n_3$ containing zero available resources, while all \NL s maintain infinite bandwidth capacity.
The threshold $Q$ is set as $Q = 2K$.
The above construction can be easily verified as polynomial-time computable.

We now establish the bi-directional reduction.
If there exists a balanced bipartition $(V_1, V_2)$ where $||V_1|-|V_2|| \leq 1$ and $|Cut(V_1, V_2)| \leq K$, then a feasible solution for ($\mathbb{P}2$)-D is constructed by the following mapping:
map all \SF s in $V_1$ to $n_1$, where $|V_1|$ equals the capacity of $n_1$; map all \SF s in $V_2$ to $n_2$, where $|V_2|$ equals the capacity of $n_2$.
Consequently, the \LL s corresponding to $|Cut(V_1, V_2)|$ have to be carried through the \NL s $l^s_{n_1 n_2}$ and $l^s_{n_2 n_3}$.
The LLnM cost of the above solution is $\tilde{Q} = 2\tilde{K} \leq 2K = Q$.

Conversely, if there exists a feasible transformed solution to ($\mathbb{P}2$)-D with cost $\tilde{Q} \leq Q$, then since $n_3$ has zero capacity, \nolinebreak all \SF s must map to either $n_1$ or $n_2$.
Let $V_1$ and $V_2$ be the sets of nodes mapped to $n_1$ and $n_2$, respectively.
Here, the balanced condition $||V_1|-|V_2|| \leq 1$ is automatically satisfied by the capacity constraints of \CN s.
Moreover, since the LLnM process involves two \NL s, $\tilde{K} = \tilde{Q}/2 \leq Q/2 = K$.
Thus, a GBP-D problem can be reduced to a ($\mathbb{P}2$)-D problem.

Since the GBP-D problem is NP-hard, the ($\mathbb{P}2$)-D problem is also NP-hard.
Note that ($\mathbb{P}2$)-D is the original ($\mathbb{P}2$)'s decision version, which is believed to be less intractable than the latter.
Eventually, we prove that ($\mathbb{P}2$) is NP-hard.
\end{proof}
\begin{corollary}
	\textup{($\mathbb{P}1$)} is NP-hard.
    \label{c1}
\end{corollary}
\begin{proof}
    Since ($\mathbb{P}2$) can be reduced to ($\mathbb{P}1$) in polynomial time, and ($\mathbb{P}2$) is NP-hard by Theorem \ref{T1}, ($\mathbb{P}1$) is NP-hard.
\end{proof}


\subsection{Practical Challenges of SEM}
While Theorem \ref{T1} establishes the NP-hardness of ($\mathbb{P}1$), its practical resolution is further complicated by the new characteristics and associated challenges introduced by explicitly optimizing \SF \ co-location in the SEM problem.
We illustrate these complexities with the following example.

Consider the scenario in Fig.~\ref{fig: case}, which involves two SEs representing incoming requests and examines how mapping decisions for the first request affect the feasibility of mapping the second. 
\SF s are labeled with lowercase letters, with their computing requirements shown inside each node. 
Each \LL\ requires 1 unit of bandwidth. 
The CPN environment contains four \CN s (labeled by capital letters), with node and link resource capacities indicated by the adjacent numbers.

Since co-located \SF s jointly affect each \CN's remaining resources and the total bandwidth requirement of \LL s that must be mapped onto tunnels (\textit{Cut}-\LL s), node-level decisions alone cannot capture the global impact of mapping. 
To characterize the complexity and intertwined global effects of such decisions, we introduce three high-level observation perspectives that explicitly reflect the collective influence of co-located \SF s, as shown in Fig.~\ref{fig: case}:

\textit{(a)} \textit{Co-location group, which describes how \SF s are grouped together based on their aggregate computing requirements within an SE (\emph{e.g.}, 8:2 and 7:3 in the example);}

\textit{(b)} \textit{Co-location group assignment, which specifies the particular \SF \ members for a given group (\emph{e.g.}, \{a,b,d\}\&\{c\} for 8:2 and \{a,b,c\}\&\{d\} for 7:3); and}

\textit{(c)} \textit{Co-location group mapping, which determines how each co-location group is mapped onto \CN s (\emph{e.g.}, E:8 \& F:2 for 8:2 and E:7 \& G:3 for 7:3).}

Co-location group, mapping, and assignment choices have different impacts on the computing and network resources of the CPN environment. 
The choice of co-location group (\emph{e.g., 7:3 vs. 8:2}) affects computational utilization in CPN; for example, a 7:3 grouping inevitably leaves 1 unit of computing fragment at \CN\ \emph{E}, regardless of group assignment or mapping. 
For the same group assignment (\emph{e.g., \{b,c,d\}~\&~\{a\}}), different group mappings (\emph{E:8~\&~F:2 vs. E:8~\&~G:2}) result in different \NL\ usages; mapping to \emph{G} instead of \emph{F} adds \emph{F--G} link usage, making it a bottleneck. 
With the same group mapping (\emph{e.g., E:8 \& F:2}), different group assignments (\emph{\{a,b,d\}~\&~\{c\} vs. \{b,c,d\}~\&~\{a\}}) change the bandwidth requirements across all \textit{Cut}-\LL s; for instance, \emph{\{a,b,d\}~\&~\{c\}} increases the requirement from 2 to 3 and causes the \emph{E--F} link to become a bottleneck. 
All these scenarios may block subsequent requests. 
These observations highlight the practical challenges in solving the SEM problem and offer valuable design insights.

\begin{figure*}[!tp]
    \centering
    \includegraphics[width=0.73\linewidth]{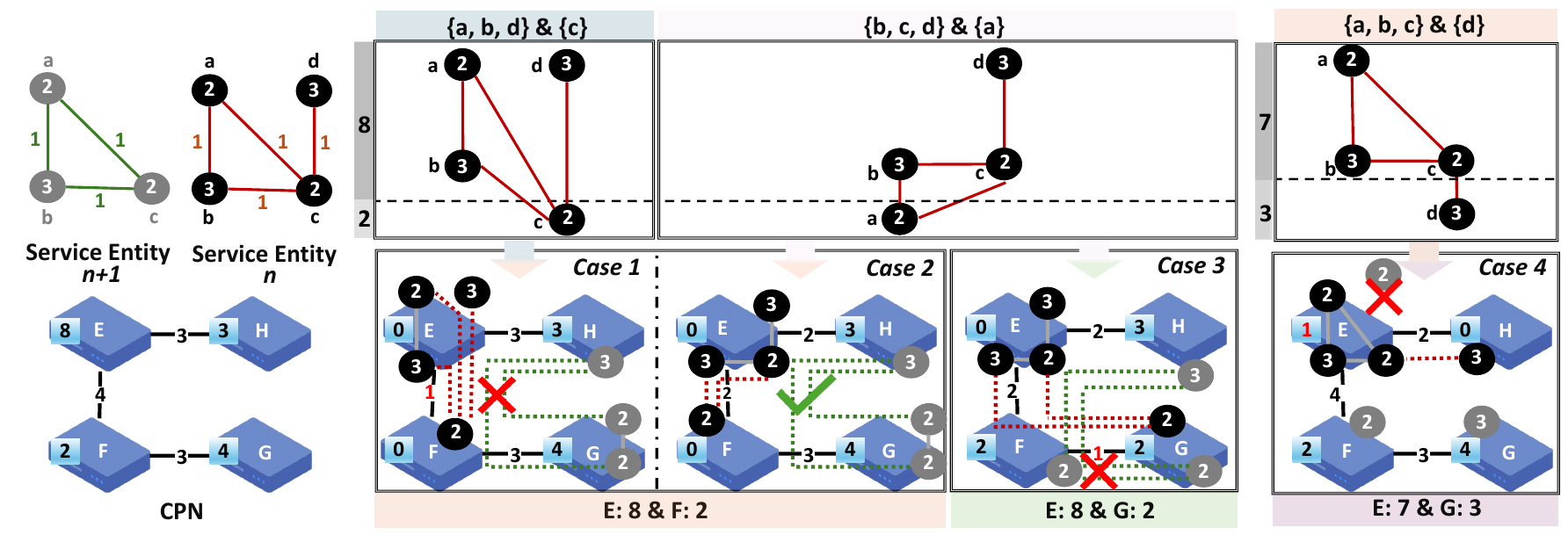}
    \caption{Examples illustrating the additional complexities and challenges of the SEM problem introduced by SF co-location. The mapping solutions that lead to the rejection of the subsequent request are marked with red crosses, and the resource bottleneck locations causing the issue are indicated with red numbers.}
    \label{fig: case}
\end{figure*}

     
\emph{Challenge 1:} 
Co-location induces strong coupling among SFnM variables with cascading effects: it can pack \SF\ demands into supply-matching patterns, eliminate \LL s, and shorten \textit{Cut}-\LL\ paths, making node-wise optimization inadequate (e.g., Cases 4, 1, and 3 as counterexamples of ignoring such coupling, respectively).

\emph{\textbf{Insight 1: Reformulate the problem via graph partitioning~\cite{metis}.}}
Viewing co-location groups as partition weights makes the couplings explicit: the weights determine resource matching, the partitioning optimizes group assignment by minimizing cut bandwidth, and the subgraph placement controls mapping paths (Section~\ref{subsecs:analysis-A}).

\emph{Challenge 2:} Co-location expands the solution space explosively (from $P(|N^s|,|N^v|)=\frac{|N^s|!}{(|N^s|-|N^v|)!}$ to $|N^v|^{|N^s|}$ when $|N^v| \leq |N^s|$), making global optimality harder to achieve and local search prone to local optima.

\emph{\textbf{Insight 2: Enable adaptive search via explicit coupling awareness.}} The explicit coupling from \textbf{\textit{Insight 1}} enables joint search over groupings and mappings; combined with global metrics and prioritizing assignments that minimize \textit{Cut}-\LL\ bandwidth, such search avoids premature convergence and effectively explores the solution space (Section~\ref{subsecs:analysis-B}).

\emph{Challenge 3:} The enlarged solution space and strong coupling hinder effective global evaluation; local metrics~\cite{divide} often lead to suboptimal results (e.g., Case 4 minimizes bandwidth cost of \textit{Cut}-\LL\ mapping but is not globally optimal).

\emph{\textbf{Insight 3: Assess global impact using resource fragmentation.}}
Resource fragmentation caused by bottlenecks (Cases 1, 3, and 4) can serve as a proxy for the global effect of each decision. Quantifying such fragmentation guides the search toward better solutions (Section~\ref{subsecs:analysis-C}).

\section{
ABS: Framework Design and Implementation
}
\label{sec: framework}

Based on the identified insights, this section reformulates the problem, designs the search framework, and develops a global evaluation mechanism.

\subsection{Reformulation for Coupling-Aware Decision-Making}
\label{subsecs:analysis-A}


In addition to the coupling among SFnM variables, a major challenge in SEM arises from the deep interdependence between SFnM and LLnM. 
On one hand, as shown in Fig.~\ref{fig: case}, co-location decisions shape the bandwidth needs and tunnel options for \textit{Cut}-LLs, directly affecting LLnM bandwidth costs as in (\ref{f8}).
On the other hand, suboptimal \textit{Cut}-LL mappings in LLnM can fragment \CN\ resources: network resource depletion or uneven allocation may leave computing resources scattered across \CN s without sufficient network capacity to interconnect them for future SEs.
This tight coupling complicates joint optimization across both processes.

Therefore, we adopt a local greedy approach based on the ($\mathbb{P}2\text{-a}$) problem, decomposing SEM into a more tractable two-phase structure.
In phase one, we solve the minimum-cut \SF\ mapping (mcuSm) for SFnM, identifying co-location groups, assignments, and mappings that minimize the bandwidth requirements of \textit{Cut}-LLs.
In phase two, we solve the minimum-cost \LL\ mapping (mcoLm) for LLnM, using the SFnM solution to minimize the additional bandwidth costs from mapping these \textit{Cut}-LLs.
While this dual-phase method cannot guarantee global optimality due to the coupling between SFnM and LLnM, it provides a practical approach to decouple the problem and focus on the complex interactions among SFnM decision variables.


\subsubsection{The mcuSm optimization}
This optimization problem is formulated as:
\begin{subequations}
	\begin{align}
		\mathbb{P}&3: \min_{\mathbf{x}} \Phi(\mathbf x) \triangleq
		\sum_{l^v_{uw} \in L^v_i} b(l^v_{uw})
		\sum_{m^s \in N^s} x^{u^v}_{m^s} (1 - x^{w^v}_{m^s}) \\
		&s.t. \quad (\ref{f1}) - (\ref{f3}).
		\end{align}
\end{subequations}

Although ($\mathbb{P}3$) is NP-hard (Corollary \ref{coro3} in Appendix~\ref{apdxA}), this formulation provides a foundation to systematically achieve coupling-aware decision making.
Guided by \emph{\textbf{Insight 1}}, we introduce a tailored graph-partitioning formulation~\cite{metis} that yields an equivalent reformulation of ($\mathbb{P}3$), thereby transforming SFnM from independent per-\SF\ placement decisions to explicit co-location-aware decision making.

\begin{definition}[Proportional Weight-Constrained \textit{k}-way Graph Partitioning Problem, PW-kGPP]
	Consider an undirected graph $G = (V, E)$, where each vertex $i \in V$ has a vertex weight $w_i > 0$, each edge $(i,j) \in E$ has an edge weight $c_{ij} > 0$, \nolinebreak and a normalized weight proportion set $PS = \{r_1, r_2, ..., r_k\}$ satisfies $\sum^k_{p=1} r_p = 1$.
	The problem requires partitioning $V$ into $k$ disjoint subsets $V_1, V_2, ..., V_{k}$, such that the sizes of the subsets correspond to the given proportion weights, while minimizing the total weight of the cut edges (edges connecting different subsets). The formal formulation is provided in Appendix~\ref{apdxc}.
\end{definition}
According to Definition \ref{D3}, we reformulate the ($\mathbb{P}3$) problem and propose the problem ($\mathbb{P}4$):
\begin{subequations}
	\begin{align}
		\label{f19a}
		\mathbb{P}&4: \min_{\boldsymbol{\rho},\mathbf{x},\mathbf{z}}
        \Phi(\mathbf x) \\
		\label{f19b}
		&s.t. \quad (\ref{f1}) - (\ref{f3}), \\
		\label{f19c}
		& \quad \forall m^s \in N^s: \nonumber \\
		& \quad (1-\theta)\rho_{m^s}
		\leq \frac{\sum_{u^v \in N^v_i} c(u^v) x^{u^v}_{m^s}}{\sum_{u^v \in N^v_i} c(u^v)}
		\leq (1+\theta)\rho_{m^s}, \\
		\label{f19d}
		& \quad \forall u^v \in N_i^v, \forall m^s \in N^s: \quad
		 x^{u^v}_{m^s} \le z_{m^s}, \\
        \label{f19d2}
        & \quad \forall m^s \in N^s: \quad \rho_{m^s} \le z_{m^s}, \\
        \label{f19d3}
        & \quad \forall m^s \in N^s: \quad z_{m^s} \in \{0,1\}, \\
        \label{f19f}
        & \quad \forall m^s \in N^s: \rho_{m^s} \ge 0, \\
		\label{f19e}  
		& \quad \sum_{m^s \in N^s} \rho_{m^s} = 1.
	\end{align}
\end{subequations}

The optimization objective of ($\mathbb{P}4$) is the same as that of ($\mathbb{P}3$), and the binary decision variable $x^{u^v}_{m^s}$ characterizes the standard SFnM process.
In addition, we introduce a continuous variable $\rho_{m^s}$ to denote the fraction of the total \SF\ resource demand accounted for by the \SF s hosted on \CN\ $m^s$.
We also introduce an auxiliary binary variable $z_{m^s}$ indicating whether \CN\ $m^s$ hosts at least one \SF.
The parameter $\theta$ is a configurable tolerance coefficient.
Furthermore, a ``subpart'' of an SE is defined as the combination of \SF s and the \LL s between them that are assigned together to a single \CN.

Similar to (\ref{f18b}) (in Appendix~\ref{apdxc}), constraint (\ref{f19c}) restricts which \SF s can be grouped on each \CN\ $m^s$ to form a subpart, while bounding the deviation between the actual resource proportion assigned to $m^s$ and the specified $\rho_{m^s}$ within tolerance $\theta$.
Constraints (\ref{f19d})--(\ref{f19d3}) couple $\boldsymbol{\rho}$ and $\mathbf{x}$ via $z_{m^s}$: (\ref{f19d}) allows an \SF\ to be mapped to \CN\ $m^s$ only when $z_{m^s}=1$, whereas (\ref{f19d2}) and (\ref{f19d3}) ensure that $\rho_{m^s}=0$ if and only if $z_{m^s}=0$.
Together, these constraints enforce that if $\rho_{m^s}=0$, then $x^{u^v}_{m^s}=0$ for all $u^v$, ensuring the intended logical consistency between proportion variables and mapping decisions.
Finally, (\ref{f19f}) and (\ref{f19e}) ensure that $\boldsymbol{\rho}$ is a valid proportion vector.

According to the definitions, and as illustrated in the examples in Fig.~\ref{fig: case}, it can be observed that under the given constraints, $\rho_{m^s}$ jointly determines the co-location group and its mapping, while $x^{u^v}_{m^s}$ determines the assignment of the co-location group.
The auxiliary variable $\mathbf{z}$ only serves to express the logical relationship between $\boldsymbol{\rho}$ and $\mathbf{x}$ and does not change the objective of the original SFnM problem.
Therefore, ($\mathbb{P}4$) provides a structured formulation that captures the intricate couplings between SFnM variables and explicitly optimizes \SF\ co-location. 
As shown in Appendix~\ref{apdxA} (Theorem~\ref{t2}), ($\mathbb{P}3$) and ($\mathbb{P}4$) are equivalent.

Unfortunately, both ($\mathbb{P}3$) and ($\mathbb{P}4$) are still NP-hard problems. However, we observe that if a set of feasible proportions $\boldsymbol{\rho}$ is given, ($\mathbb{P}4$) reduces to the PW-kGPP problem. 
In this reduction, non-zero values of $\boldsymbol{\rho}$ in ($\mathbb{P}4$) correspond to the weight proportion set $PS$ in PW-kGPP, with the number of non-zero elements equal to $k$, and other correspondences following directly.

The NP-hardness of PW-kGPP is established in Appendix~\ref{apdxA} (Theorem~\ref{t3}), which directly implies that ($\mathbb{P}4$) is NP-hard (Corollary~\ref{coro2}) and hence ($\mathbb{P}3$) is NP-hard (Corollary~\ref{coro3}).


\subsubsection{The mcoLm optimization}
After the mcuSm optimization, the set of \textit{Cut}-\LL s, denoted as $CL_i^v$, is determined. 
The next step is to map each $l^v_{uw} \in CL_i^v$ onto a tunnel, with the objective of minimizing the total bandwidth cost. 
This mapping subproblem can be reduced to an instance of the Integer Multicommodity Flow Problem (IMCF), which is well known to be NP-hard~\cite{mcf}.

\subsection{
Bilevel optimization framework for adaptive searching
}
\label{subsecs:analysis-B}
Inspired by \emph{\textbf{Insight 2}}, we reformulate ($\mathbb{P}4$) as a bilevel optimization problem ($\mathbb{P}4\text{-BOP}$), where the upper level searches over the proportion weight vector (PWV) $\boldsymbol{\rho}$ and the lower level solves the induced combinatorial mapping problem.
For any fixed $\boldsymbol{\rho}$, the lower level specifies a best-response mapping decision by solving the induced subproblem~\cite{bilevel}.
\begin{subequations}
\begin{IEEEeqnarray}{ll}
\label{p4bop}
&(\mathbb{P}4\text{-BOP}):\\ \nonumber
& \text{Upper-level:}\ \min_{\boldsymbol{\rho}}\ \Phi(\mathbf{x}^*(\boldsymbol{\rho})) \\
& \qquad\qquad\ \text{s.t.}\ (\ref{f19f})-(\ref{f19e}) \\
& \text{Lower-level:}\ (\mathbf{x}^*(\boldsymbol{\rho}),\mathbf{z}^*(\boldsymbol{\rho})) \in 
\arg\min_{\mathbf{x},\mathbf{z}}\ \Phi(\mathbf{x}) \\
& \qquad\qquad\ \text{s.t.}\ (\ref{f1})-(\ref{f3}),(\ref{f19c}) -(\ref{f19d3}).
\end{IEEEeqnarray}
\end{subequations}

As shown in Appendix~\ref{apdxA} (Theorem~\ref{t4}), ($\mathbb{P}4$) and ($\mathbb{P}4\text{-BOP}$) are equivalent in the sense that they attain the same globally optimal objective value, and their optimal solutions coincide when projected onto the original decision variables.
Thus, we can use this bilevel form as a decomposition interface for the ensuing nested solution procedure.

Taking a closer look, the aforementioned dual-phase structure suggests solving the mcuSm optimization first to determine the \textit{Cut}-\LL\ set, and then solving the mcoLm optimization (i.e., the IMCF problem) conditioned on the selected \textit{Cut}-\LL s.
Moreover, the bilevel structure ($\mathbb{P}4\text{-BOP}$) induces a two-stage routine for the mcuSm optimization: we search for the upper-level variable $\boldsymbol{\rho}$, and for each fixed $\boldsymbol{\rho}$, the induced lower-level subproblem can be reduced to a PW-kGPP instance.
Through this series of decomposition and reformulation, we obtain a practical search framework for $(\mathbb{P}2)$: the search can be reduced to navigating the low-dimensional $\boldsymbol \rho$-space, where each candidate $\boldsymbol{\rho}$ induces a PW-kGPP and an IMCF subproblem that can be solved to construct a feasible solution $(\mathbf{x},\mathbf{f})$, with iterative feedback guiding the upper-level exploration.
Accordingly, we state the following proposition.
\begin{proposition}
	Assume that: (i) the outer layer performs an exact global search over the feasible set of $\boldsymbol{\rho}$ for the upper-level search; and
    (ii) for any fixed $\boldsymbol{\rho}$, the inner layer solves the corresponding PW-kGPP and IMCF subproblems to global optimality.
    Then, the solution returned by the nested procedure is a global optimal solution to ($\mathbb{P}2$).
    The proof is given in Appendix~\ref{apdxA}.
	\label{T5}
\end{proposition}


We emphasize that Proposition~\ref{T5} characterizes an idealized setting in which $(\mathbb{P}2)$ is solved optimally by exactly minimizing its objective (\ref{f17}a) over $\boldsymbol{\rho}$.
In practice, both PW-kGPP and IMCF are intractable; consequently, the nested procedure generally returns a high-quality but suboptimal solution to $(\mathbb{P}2)$. 
Encouragingly, efficient algorithms exist for related special cases, and practical heuristics are available: the balanced graph partitioning problem (a special case of PW-kGPP) admits an $O(\log n)$ approximation~\cite{balanced}, heuristics such as METIS~\cite{metis} perform well empirically for general PW-kGPP instances, and for IMCF, the $k$-shortest path algorithm provides high-quality candidate paths for each \textit{Cut}-LL.

However, since $(\mathbb{P}2)$ optimizes in a greedy, per-request manner, solving it independently can be myopic with respect to the cumulative online objective in $(\mathbb{P}1)$.
To mitigate this limitation, the outer-loop update of $\boldsymbol{\rho}$ should be guided by a global evaluation metric that aggregates performance over the entire online process and serves as a proxy for the objective in $(\mathbb{P}1)$.
By capturing cross-request effects (e.g., resource fragmentation), this metric discourages per-request choices that appear reasonable for $(\mathbb{P}2)$ but degrade future feasibility emphasized by $(\mathbb{P}1)$.
Thus, the outer loop search over $\boldsymbol{\rho}$ should not be driven by minimizing the $(\mathbb{P}2)$ objective; instead, it should aim to optimize this proxy metric.
Meanwhile, the inner subproblem solvers remain responsible for constructing a feasible SEM decision for each fixed $\boldsymbol{\rho}$.

Motivated by the above decomposition and the need for global guidance across time, we propose the modular ABS framework (Fig.~\ref{fig: ABS}), which couples time-global guidance with per-request local optimization for SEM.
It consists of three interdependent core components that form an iterative feedback loop, enabling dynamic adjustments and continual improvement in terms of the long-term evaluation metric:
\begin{itemize}
    \item Adaptive Search Engine:
    It employs adaptive strategies to traverse the PWV space in search of a high-performing $\boldsymbol{\rho}$ candidate. 
    Based on feedback from the Global Evaluation Module, it prioritizes regions with high scores and dynamically refines its search trajectory.
    
    \item Local Subproblem Solver:
     It efficiently solves the PW-kGPP and IMCF subproblems using state-of-the-art practical algorithms, producing high-quality SEM decisions for each candidate $\boldsymbol{\rho}$.
    
    \item Global Evaluation Module:
    It assesses the quality of the SEM solution from a time-global perspective and provides feedback to the Adaptive Search Engine, guiding its exploration toward promising regions.
\end{itemize}
\begin{figure}[!tp]
    \centering
    \includegraphics[width=0.85\linewidth]{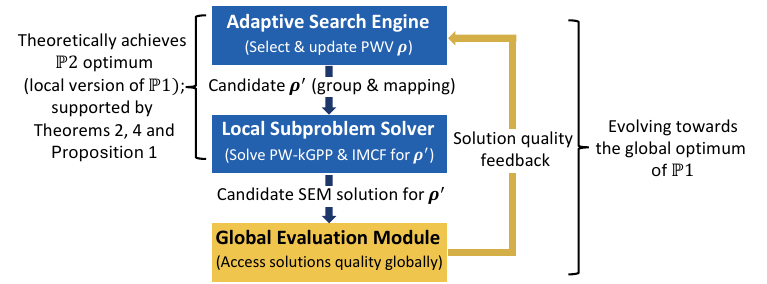}
    \caption{Architecture of the Adaptive Bilevel Search (ABS) framework for the SEM problem.}
    \label{fig: ABS}
    \vspace{-10pt}
\end{figure}

From a high-level perspective, ABS bridges global (online) evaluation and local (per-request) optimization: it uses cumulative performance feedback (a proxy of the objective in $(\mathbb{P}1)$) to adaptively steer the outer-loop search, while retaining a best-response inner solver that, for each fixed $\boldsymbol{\rho}$, efficiently solves the induced PW-kGPP and IMCF subproblems.


\subsection{Fragmentation-Based Metrics for Global Evaluation}
\label{subsecs:analysis-C}

According to \emph{\textbf{Insight 3}}, it is necessary to evaluate each mapping decision found by ABS in real time based on the current fragmentation state, thereby further enhancing its iterative capability toward the long-term optimum.
To overcome the limitations of conventional load-dependent and resource-centric fragmentation evaluation methods~\cite{beware}, we propose a service-centric approach that focuses on the immediate impact of each mapping decision on resource utilization, without relying on any prior knowledge of input service entities.

We first provide additional definitions for the evaluated mapping decision $(\mathbf{\hat x, \hat f})$ to facilitate expression. 
Let $N_i^s$ denote the set of all \CN s participating in hosting \SF s in the decision $(\mathbf{\hat x, \hat f})$\footnote{The mapping decision in planning is defined as the one for the $i^{th}$ request.}.
Then, \nolinebreak the resource amount utilized of any participating \CN , $m^s \in N_i^s$, can be formulated as:
\begin{align}
	P_C(m^s) = \sum_{u^v \in N^v_i} \hat x^{u^v}_{m^s} \cdot c(u^v).
\end{align}
Besides, we denote each participating \CN 's correlated bandwidth consumption as:
\begin{align}
	P_{BW}(m^s) = \sum_{l^v_{uw} \in \hat {CL}^v_i} (\hat x^{u^v}_{m^s} + \hat x^{w^v}_{m^s}) \cdot b(l^v_{uw}).
\end{align} 
Observing each of the three fragmentation scenarios exemplified in Fig.~\ref{fig: case}, we propose tailored evaluation metrics designed to prevent the formation of such fragments, where a higher metric value indicates a more desirable solution.

\textbf{Node Resource Exhaustion Degree (NRED)} quantifies the overall resource exhaustion level across all participating \CN s.
Inspired by Case 4, we establish a ``using it all'' criterion for each decision. 
Once a \CN \ is selected to host \SF s, \SF s are mapped to consume as much of its available resources as possible, prioritizing their total resource requirements over their mere count. 
This approach minimizes residual computing resources on utilized nodes, thereby reducing potential resource fragmentation.
It can be formulated as follows: 
\begin{align}
	NRED =
	\frac{\sum_{m^s \in N_i^s} \biggl( \frac{P_C(m^s)}{C(m^s)} \biggr)}{\sum_{m^s \in N_i^s} \biggl \lceil \max \bigl( 1 - \frac{P_C(m^s)}{C(m^s)} - \delta, \, 0 \bigr) \biggr \rceil + \epsilon},
\end{align}
where $\delta$ is a threshold parameter, and $\epsilon$ is a small constant preventing zero division.
The numerator sums the resource utilization ratio, $\frac{P_C(m^s)}{C(m^s)}$, of each participating \CN \ to reward node exhaustion, increasing as each node approaches full utilization. 
Meanwhile, the denominator equally penalizes the \CN s that have not fully utilized their resources, where $1 - \frac{P_C(m^s)}{C(m^s)} > 0$, assigning a uniform penalty factor. 
The NRED value decreases as the number of under-utilized \CN s increases.
Moreover, the threshold $\delta$ excludes those nearly exhausted \CN s (remaining ratio: $1 - \frac{P_C(m^s)}{C(m^s)} \leq \delta$) from penalties to enhance practical applicability.
The optimal NRED value is achieved when all participating \CN s are fully exhausted, leading to a denominator of $\epsilon$ and maximizing NRED. 

\textbf{Computing to Bandwidth Utilization Gap (CBUG)} evaluates the overall computing-to-bandwidth utilization status of all participating \CN s.
In Case 1, \CN \ \emph{E} does not have sufficient bandwidth to support the \LL s \emph{b-a} and \emph{b-c} in the next request.
To avoid this situation, we expect that after each mapping, the remaining computing resources of a node are minimized, while its residual correlated bandwidth is maximized, thereby allowing the node to continue functioning as a forwarding node.
Here, a \CN's correlated bandwidth refers to the total available bandwidth of all \NL s connected to it.
We propose a ``keeping the gap'' bandwidth utilization strategy to achieve this purpose.
In this design, the correlated bandwidth consumption should be as small as possible relative to the computing resource utilization for each participating \CN \ $m^s$. 
That is, the ratio $\frac{P_C(m^s)}{P_{BW}(m^s)}$, which we refer to as the utilization gap, should be maximized. 
Furthermore, a larger gap indicates that more computing resources are utilized per unit of bandwidth consumed, thereby avoiding cases where bandwidth is exhausted while computing resources remain underutilized.
\textbf{CBUG} is formulated as follows:
\begin{align}
	CBUG = 
	\frac{\sum_{m^s \in N_i^s} \frac{P_C(m^s)}{P_{BW}(m^s) + \epsilon}}{| N_i^s |},
\end{align}
where we set $\epsilon$ as a small constant to prevent zero division.

\textbf{Path Node Valuelessness Level (PNVL)} quantifies the impact of mapping \textit{Cut}-\LL s on computing resource utilization. 
During LLnM, a \CN\ consumes bandwidth both for locally hosted \SF s’ \textit{Cut}-\LL s and, as a forwarding node, for tunnel paths. After SFnM, the residual computing resources of each \CN\ are known; nodes with more resources are more valuable for future requests and should ideally be less involved in forwarding. 
For example, in Case 3, although CN \emph{F} retains sufficient computing resources to host a future \SF, it lacks the necessary bandwidth for associated \LL s due to prior forwarding tasks. 
Therefore, we adopt a ``traversing less valuable nodes'' criterion and define a metric to evaluate the valuelessness level of forwarding nodes along the mapping path of each \textit{Cut}-\LL. 
Specifically, for a \textit{Cut}-\LL\ $l^v_{uw}$, $\hat{MoP}(l^v_{uw})$ denotes the set of \CN s acting as its forwarding nodes under decision $(\mathbf{\hat x, \hat f})$, and the metric is formulated as:
\begin{align}
	P_{PV} (l^v_{uw}) = 
	\frac{\sum_{m^s \in MoP(l^v_{uw})}
\frac
{b(l^v_{uw})}
{C(m^s) - P_C (m^s) + \epsilon}}
{e^{- |\hat {MoP}(l^v_{uw})|}},
\end{align}
where $\epsilon$ is a small constant to prevent zero division.
In the numerator, we sum the ratios of each \textit{Cut}-\LL 's bandwidth requirement to the residual computing resources of each forwarding node, discouraging the traversal of more valuable \CN s. The denominator uses an exponential function to penalize paths with excessive hop counts, thereby limiting the number of forwarding nodes and reducing bandwidth consumption during LLnM.
Based on this, we define \textbf{PNVL} to assess the overall impact of all \textit{Cut}-\LL s in a decision:
\begin{align}
PNVL =
\frac{\sum_{l^v_{uw} \in \hat {CL}i^v} P_{PV}(l^v_{uw})+\epsilon'}{|\hat {CL}_i^v|+\epsilon},
\end{align}
where $\epsilon \ll \epsilon'$ prevent division by zero. This metric is the arithmetic mean of path evaluation across all \textit{Cut}-\LL s.

\subsection{
Implementation Details
}
\label{secs:algorithm}

We now present the ABS implementation via Distributed Particle Swarm Optimization, including the workflow, evolution, and initialization strategies, as well as a computational complexity analysis (see the pseudocode of Algorithms 1–4 in Appendix~\ref{apdxb}).

\subsubsection{Basic Idea} To implement the ABS framework, we adopt the distributed elite-guided learning particle swarm optimizer (DEGLSO) framework~\cite{distributed}, which offers inherent scalability and flexibility for complex CPN scenarios.
It builds on the classic Particle Swarm Optimization (PSO) algorithm by incorporating elite-guided learning and distributed computing strategies, enhancing both convergence speed and scalability.

From the perspective of ABS's architecture, we adopt the DEGLSO as the Adaptive Search Engine. 
We integrate our proposed fragmentation-based evaluation mechanism into constructing the fitness function of PSO, serving as the Global Evaluation Module, which is formulated as: 
\begin{align}
	\mathcal{F} = \frac{1}{
	\varpi_1 \cdot NRED + \varpi_2 \cdot CBUG + \varpi_3 \cdot PNVL
	},
	\label{f29.5}
\end{align}
where $\varpi_1$, $\varpi_2$, and $\varpi_3$ represent the weighting coefficients among the evaluation metrics.
The specific values of these weights can be determined based on the observed correlation between higher values of these metrics and improved overall performance, as discussed in Section~\ref{subsubsec:extends}.
To solve the two subproblems, PW-kGPP and IMCF, we utilize METIS\footnote{https://github.com/KarypisLab/METIS} and k-shortest path\footnote{https://networkx.org/algorithms/shortest\_paths.html} as the Local Subproblem Solver.

However, DEGLSO was initially designed for general optimization problems, and its mechanisms do not fully align with the specific requirements of the ABS framework. 
The following subsections detail how we adapt and apply DEGLSO in our implementation.


\subsubsection{Overall Workflow}


     

\begin{figure}[!tp]
    \centering

    \includegraphics[width=0.8\linewidth]{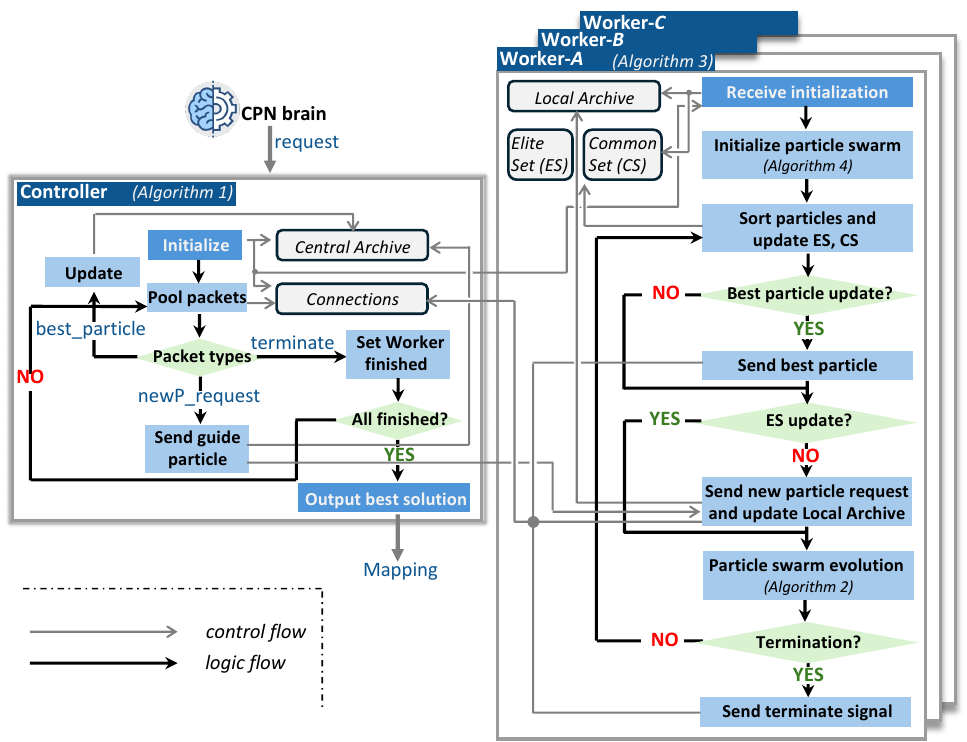}
     
    \caption{ABS framework implementation scheme.}
    \label{fig: algorithm}
    \vspace{-10pt}
\end{figure}

Our \OURALGO\ inherits the controller-worker model of DEGLSO, where the main procedure functions as the controller.
During optimization, the controller aggregates and updates the global search state based on feedback from all the workers.
Each worker independently manages a local particle swarm, performing iterative evolution updates under the guidance of the controller.
After all workers complete their iterations, the controller selects the particle with the best fitness value among all reported results and uses its corresponding solution as the mapping decision for the current service request.
If the workers cannot report any particles, this request will be rejected  (see Fig.~\ref{fig: algorithm} for an overview).

Unlike general PSO approaches, where a particle's position directly represents a candidate solution, our implementation assigns the position a dual role.
First, it explicitly encodes $\boldsymbol{\rho}$, the core variable updated during the search (referred to as the explicit position, or simply the position).
Second, since fitness evaluation requires the SEM solution $(\mathbf{x}, \mathbf{f})$, this solution is derived from the explicit position $\boldsymbol{\rho}$ and forms the particle's implicit position (hereafter, the solution).
Positions are updated via (\ref{f30}) and (\ref{f31}), while fitness is computed on the solution using $\mathcal{F}(\cdot)$ in (\ref{f29.5}).

The controller maintains an archive as the core of the collaborative process, storing the current optimal particles reported by workers. 
This archive continuously incorporates the latest global search results and provides guidance for workers’ evolutionary updates. 
Controller–worker collaboration is achieved through asynchronous communication: the controller establishes dedicated channels with each worker, monitors incoming messages, updates the archive, and responds with evolutionary guidance as needed. 
Once all workers have completed their iterations, the controller finalizes and returns the results.
During the optimization process, communication occurs when a worker discovers a better local solution, requires guidance for further search, or finishes its computation. 
In each case, the controller evaluates new information, updates the archive if appropriate, and provides selected particles from the archive to guide the workers’ evolution (see details in Fig.~\ref{fig: algorithm} and Algorithm~\ref{alg:master}).

\subsubsection{Evolution Iteration Phase}
\label{subsubsecs: evolve}
During this phase, each worker in \OURALGO \ maintains a local particle swarm and executes its evolutionary updates, driving the whole iteration process (see Fig.~\ref{fig: algorithm} for an overview).

To enhance high-dimensional search efficiency, we incorporate the elite-guided learning strategy of DEGLSO. 
In each iteration, the local swarm is divided into an elite set ($\mathbf{ES}$) and a common set ($\mathbf{CS}$) based on solution quality. 
Only particles in $\mathbf{CS}$ are updated, guided by both local elite particles (from the local archive, $\mathbf{LA}$) and historical best particles identified by different workers. 
This approach helps maintain diversity and enables particles to escape local optima.

Taking any dimension $d \in \{1, 2, ..., |N^s|\}$ of the vector $\boldsymbol{\rho}$ as an example, the update rule for the vectors of these common particles can be expressed as:
\begin{align}
	v_d^{t+1} = \underbrace{\zeta_1 v_d^{t}}_{\text{inertia}} + \underbrace{\zeta_2 (e^t_d - \rho^t_d)}_{\text{a random elite}} + \underbrace{\phi \zeta_3 (\hat E^t_d - \rho^t_d)}_{\text{elites' mean position}}, \quad \quad \, \label{f30}
\end{align}
\begin{align}
	\rho_d^{t+1} = \rho_d^{t} + v_d^{t+1}, \quad \quad \quad \quad \quad \quad \quad \quad \quad \quad \quad \quad \quad \label{f31}
\end{align}
\begin{align}
	\hat E^t = \frac{1}{|ES^t \cup LA^t|} \cdot \sum_{\rho \in ES^t \cup LA^t} \rho \, , \quad \quad \quad \quad \quad \quad \,  \label{f32}
\end{align}
where $\boldsymbol{\rho}^t$ and $\mathbf{v}^t$ are the position and velocity of any particle in the $\mathbf{CS}^t$, respectively; $\boldsymbol{e}^t$ is the position of a randomly selected particle from $\mathbf{ES}^t \cup \mathbf{LA}^t$; $\mathbf{\hat E}^t$ is the mean position of all the local elite particles; $t$ is the index of the current iteration; and $\zeta_1, \zeta_2, \zeta_3 \sim U(0,1)$ are independent random coefficients that stochastically weight the inertia and attraction terms.
Moreover, $\phi$ is a linear adjuster to enhance the exploitation at the late stages~\cite{modifiedpso}, expressed as:
\begin{align}
	\phi^t = 1 - \frac{t}{G_{max}},
	\label{f33}
\end{align}
where $t$ is the index of the current iteration, and $G_{max}$ is the maximal iteration number. 
In our setting, high-dimensional sparsity emerges due to \SF\ co-location, especially as the number of \CN s increases, resulting in most PWV elements being zero. Such sparsity significantly complicates efficient search and reduces the likelihood of finding feasible solutions.
To address this challenge, we introduce a separate search mechanism. Specifically, each particle first performs unconstrained exploration in continuous space ($\boldsymbol{\rho}$), fully leveraging the flexibility of free search without imposing sparsity or normalization constraints. 
Subsequently, a dynamic top-$n$ masking is applied, where only the $n$ largest components of $\boldsymbol{\rho}$ are retained and normalized, yielding a sparse vector ($\boldsymbol{\rho'}$) that satisfies the simplex constraint. This sparse vector is then used as the PWV for solving the subproblems. 
If both subproblems yield feasible solutions, the result is assigned to the particle solution, and the search dimension $n$ is reduced to progressively focus the search on the most critical variables (see details in Algorithm~\ref{alg:evolve}).

To support global optimization, each worker maintains its local state and collaborates with the controller.
After initialization, the worker updates the \textbf{CS} of particles in each iteration. 
If the local best particle improves, this information is communicated to the controller. 
When no improvement occurs in the \textbf{ES}, the worker requests guidance and updates its \textbf{LA} based on feedback from the controller. 
The worker then proceeds to evolve the particles in the \textbf{CS} for the next iteration (see Fig.~\ref{fig: algorithm} and Algorithm~\ref{alg:worker} for details).

\subsubsection{Initialization Phase}
\label{subsubsecs: init}

Particle initialization must satisfy SEM structural constraints and sufficiently cover the solution space; otherwise, insufficient constraints may produce infeasible solutions and lead to many wasted iterations, while overly strict constraints can hinder exploration.
To achieve this, we propose a semi-constrained randomized breadth-first initialization method, which efficiently initializes the particle swarm for each worker and supports effective iterations.
For each particle, candidate \CN s with abundant resources are probabilistically selected and expanded in a breadth-first manner to minimize inter-node distances, thereby reducing LLnM cost and improving SFnM feasibility. 
To ensure sufficient exploration, the process dynamically expands the candidate region, allowing broader network areas to be considered until a feasible solution is found or resources are exhausted. 
See Algorithm~\ref{alg:init} for detailed steps.

\subsubsection{Time Complexity Analysis}

The overall computational complexity of the proposed distributed ABS implementation is primarily determined by the iterative main loop executed by each worker.
It has polynomial time complexity, given by $\mathcal{O}(ITR_{max}\cdot|N_S| \cdot k\cdot|N^s|\cdot(|L^s|+|N^s|\log|N^s|))$, indicating that the algorithm can efficiently scale to large and complex problems through multi-worker parallelization. 
A detailed complexity analysis is provided in the Appendix~\ref{apdxb}.

\section{
Performance Evaluation}
\label{secs:evaluation}
\begin{table}[tb]
  \centering
  \caption{Simulation parameters}
    \begin{tabular}{llll}
    \toprule
    \textbf{Parameters} & \multicolumn{2}{l}{\textbf{CPN Topology }} & \textbf{SE Topology } \\
    \midrule
          & Random & \multicolumn{1}{p{5.415em}}{Rocketfuel} &  \\
\cmidrule{2-3}    Number of Topology & 1     & 1     & 2000 \\
    Number of Nodes & 100   & 129   & 50-100 \\
    Number of Links & 500   & 363   & \textit{Random} (0.9) \\
    Node Capacity   & 400-600 & 400-600 & 1-20 \\
    Bandwidth  & 400-600 & 400-600 & 1-20 \\
    Arrival Time  & -     & -     & \textit{Poisson} (0.1) \\
    Life Time  & -     & -     & \textit{Exp} (500) \\
    \bottomrule
    \end{tabular}%
  \label{tab:parameter}%
\end{table}%

\subsection{Experimental Setup}

\subsubsection{Benchmarks}

We compare our \OURALGO\ against six algorithms spanning three mainstream categories: heuristic (\textbf{RW-BFS}~\cite{rw},\textbf{ RMD}~\cite{rmd}), meta-heuristic (\textbf{EA-PSO}~\cite{psovne}, \textbf{GA-STP}~\cite{gastp}), and learning-based (\textbf{RL-QoS}~\cite{pgcnn},\textbf{ GAL}~\cite{gal}) methods.
Note that all these algorithms are designed for the traditional VNE problem. 
For a fair comparison, we adapt them for the SEM problem by removing the one-to-one mapping constraint between \SF s and \CN s.
 Details of baselines are given in Appendix~\ref{apdxd}.

\subsubsection{Performance Metrics}

For evaluating algorithm efficacy, our primary objectives are to maximize \textbf{acceptance ratio}, \textbf{revenue}, \textbf{long-term average revenue (LT-AR)}, \textbf{profit}, and \textbf{\CN\ Resource Utilization (CU-Ratio)}. 
\textbf{Revenue to cost ratio (RC-Ratio)} and \textbf{long-term revenue to cost ratio (LT-RC-Ratio)} are also reported to provide a reference for cost efficiency, though we note their limitations as comprehensive evaluation metrics. 
Cost alone is not used as a primary metric, as, by definition, it not only reflects resource expenditure but may also indicate higher service acceptance and revenue.
The metric details are also elaborated in Appendix~\ref{apdxd}.

\subsubsection{Simulator settings}
We conduct simulations in a numerically generated environment to evaluate algorithm performance with online service requests. SE topologies arrive sequentially to a CPN topology, which is either a random Waxman model~\cite{rw} or a real-world Rocketfuel (AS6461) instance~\cite{rocket}. 
To emulate complex services and dense inter-function dependencies, SEs are large-scale with high link connectivity~\cite{divide}. Node and link resources are uniformly distributed. Request arrivals follow a Poisson process, and their lifetimes follow an exponential distribution. All specific simulation parameters are detailed in Tab.~\ref{tab:parameter}. 

The simulation experiments are conducted on a server equipped with an Intel(R) Xeon(R) Gold 5218R CPU @ 2.10GHz, an NVIDIA GeForce RTX 3080 GPU, and 252GB of RAM. The computational environment is based on Python 3.10.16, employing PyTorch 1.12.1 and PyTorch Geometric 2.2.0 as its machine learning frameworks. GPU acceleration is facilitated by driver version 515.65.01 and CUDA 11.7. 

\subsection{Performance Results and Analysis}

\begin{table*}[!tp]
  \centering
  \caption{Comprehensive evaluation of all compared algorithms. The top three performances are highlighted using \textcolor{red}{\textbf{First}}, \textcolor{purple}{\textbf{Second}}, and \textbf{Third}.}
    \begin{tabular}{llccccccc}
    \toprule
    \multicolumn{1}{r}{} & Algorithm  & Acpt. Ratio $\uparrow$ & Revenue  $\uparrow$ & LT-AR $\uparrow$ & Profit $\uparrow$ & RC-Ratio & LT-RC-Ratio \\
    \midrule
    \midrule
    \multicolumn{9}{l}{\textit{Testing Under Random Topology}} \\[0.2em]
    \hdashline[2pt/5pt]
    \\[-0.7em]
    \multirow{2}[1]{*}{\textbf{Heuristic}} & \textbf{RW-BFS}~\cite{rw} & 0.589     &  7720887    & 181175     & 1082950     & 0.839        & 0.841 \\
          & \textbf{RMD}~\cite{rmd} & 0.381     & 4258676     & 100049     & 244752     & 0.827    & 0.828 \\
    \multirow{2}[0]{*}{\textbf{Meta-Heuristic}} & \textbf{EA-PSO}~\cite{psovne} & 0.523     & 6691894     & 154194     & 103220     & 0.530        & 0.531 \\
          & \textbf{GA-STP}~\cite{gastp} & 0.416     & 4475305     & 105223     & 351589     & 0.916       & 0.931 \\
    \multirow{2}[0]{*}{\textbf{Learning-Based}} & \textbf{RL-QoS}~\cite{pgcnn} & 0.592     & 7676094     & 182007     & 571727     & 0.752    & 0.636 \\
          & \textbf{GAL}~\cite{gal} & \textbf{0.613}     & \textbf{8249720}     & \textbf{190394}     & \textbf{1305402}     & 0.865     & 0.871 \\
    \multirow{2}[1]{*}{\textbf{Our Approach}} & \textbf{ABS}\textsubscript{init by RW-BFS} & \textcolor{purple}{\textbf{0.693}}     & \textcolor{purple}{\textbf{9411934}}     & \textcolor{purple}{\textbf{218374}}     & \textcolor{purple}{\textbf{1487462}}     & 0.745          & 0.753 \\
          & \textbf{ABS} & \textcolor{red}{\textbf{0.703}}     & \textcolor{red}{\textbf{9735826}}     & \textcolor{red}{\textbf{228028}}     & \textcolor{red}{\textbf{1613916}}     & 0.752          & 0.751 \\
    \midrule
    \midrule
    \multicolumn{9}{l}{\textit{Testing Under Real-World Topology}} \\[0.2em]
    \hdashline[2pt/5pt]
    \\[-0.7em]
    \multirow{2}[1]{*}{\textbf{Heuristic}} & \textbf{RW-BFS}~\cite{rw} & 0.234     & 2343136     & 57227     & 64054     & 0.999        & 0.997 \\
          & \textbf{RMD}~\cite{rmd} & 0.036     & 350989     & 8333     & 242     & 1.072     & 1.046 \\
    \multirow{2}[0]{*}{\textbf{Meta-Heuristic}} & \textbf{EA-PSO}~\cite{psovne} & 0.135     & 1330372     & 32558     & -675     & 0.487     & 0.502 \\
          & \textbf{GA-STP}~\cite{gastp} & 0.220     & 2028733     & 47645     & 57195     & 1.198     & 1.205 \\
    \multirow{2}[0]{*}{\textbf{Learning-Based}} & \textbf{RL-QoS}~\cite{pgcnn} & -     & -     & -     & -     & -     & - \\
          & \textbf{GAL}~\cite{gal} & \textbf{0.256}     & \textbf{2578649}     & \textbf{60050}     & \textbf{91009}     & 1.083     & 1.102 \\
    \multirow{2}[1]{*}{\textbf{Our Approach}} & \textbf{ABS}\textsubscript{init by RW-BFS} & \textcolor{purple}{\textbf{0.376}}     & \textcolor{purple}{\textbf{3927901}}     & \textcolor{purple}{\textbf{97433}}     & \textcolor{purple}{\textbf{246508}}     & 0.899     & 0.878 \\
          & \textbf{ABS} & \textcolor{red}{\textbf{0.410}}     & \textcolor{red}{\textbf{4486683}}     & \textcolor{red}{\textbf{103144}}     & \textcolor{red}{\textbf{333456}}     & 0.898     & 0.912 \\
    \bottomrule
    \end{tabular}%
  \label{tab:results}%
\end{table*}%

\subsubsection{Overall Performance}
Tab.~\ref{tab:results} presents the aggregated results of all algorithms over 2000 trials for all metrics (except for the CU-Ratio) under both CPN topologies. 
Note that ABS\textsubscript{init by RW-BFS} is a variant of \OURALGO, where the \texttt{init\_solver} function in the particle swarm initialization is replaced with the RW-BFS algorithm.

As shown in Tab.~\ref{tab:results} and Fig.~\ref{fig: cur}, \OURALGO\ and its variant consistently achieve the best performance among all evaluated algorithms on the key metrics under both random and real-world topologies.
This strong performance stems from several key design choices of ours.
First, by adopting a graph partitioning perspective, ABS transforms independent mapping decisions into a co-location-aware search process, enabling effective \SF \ grouping and mapping tailored for CPN scenarios.
Moreover, its bilevel optimization structure efficiently explores the solution space via locally optimal subproblems.
Finally, the resource fragmentation-based feedback mechanism offers a global evaluation metric that dynamically guides the search toward global optimality.

Among learning-based methods, GAL stands out among all baselines due to its two-stage design: it first imitates the node ranking pattern of RW-BFS using a graph neural network, and then refines its policy through reinforcement learning.
It's worth noting that RL-QoS, although ranking fourth on the random topology, fails to converge effectively on the more network resource-constrained real-world topology.
The main issue lies in RL-QoS's use of auto-regressive RL trained from scratch, which is prone to error accumulation and struggles with sparse rewards in challenging settings.
In contrast, GAL's initial imitation learning phase enables more effective early exploration, which helps overcome these challenges.

\begin{figure*}[tp]
    \centering
    \subfigure[Acceptance ratio]{
        \includegraphics[width=0.265\linewidth]{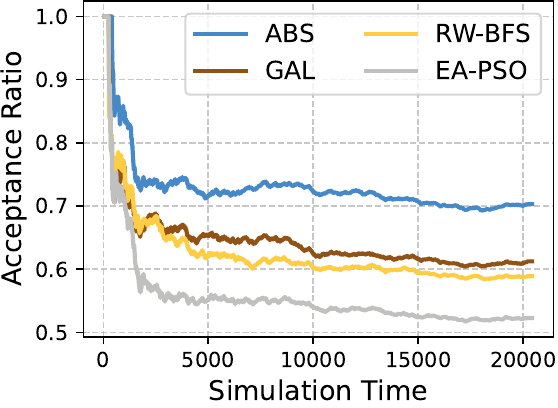}
        \label{fig: randomacpt}
    }
    \hspace{0.1em} 
    \subfigure[Long-term average revenue]{
        \includegraphics[width=0.265\linewidth]{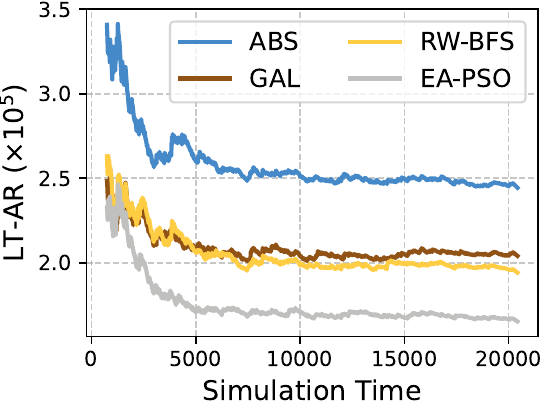}
        \label{fig: randomltr}
    }
    \hspace{0.1em}
    \subfigure[Long-term revenue to cost ratio]{
        \includegraphics[width=0.265\linewidth]{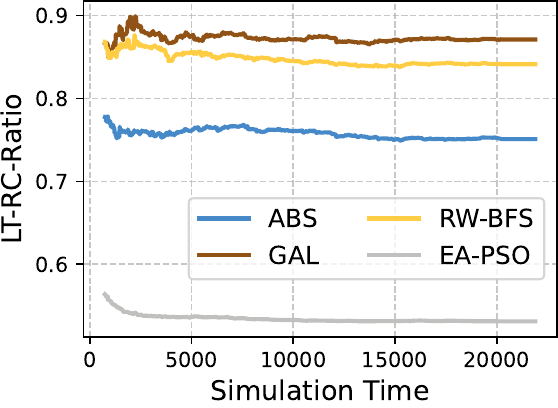}
        \label{fig: randomltrc}
    }
    \\
    \subfigure[Acceptance ratio]{
        \includegraphics[width=0.265\linewidth]{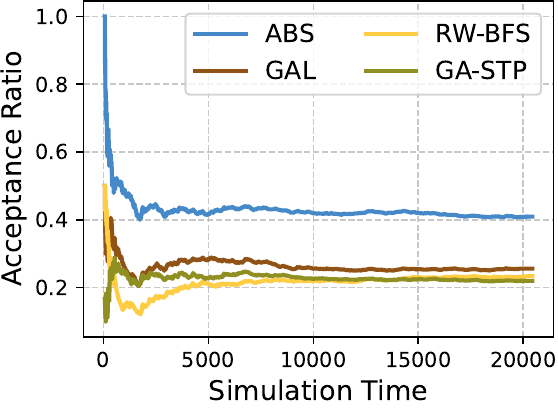}
        \label{fig: realacpt}
    }
    \hspace{0.1em} 
    \subfigure[Long-term average revenue]{
        \includegraphics[width=0.265\linewidth]{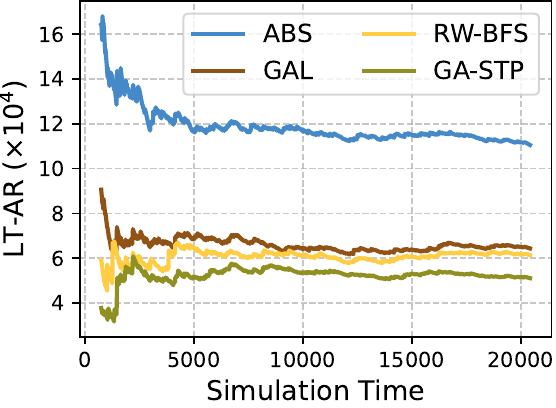}
        \label{fig: realltr}
    }
    \hspace{0.1em}
    \subfigure[Long-term revenue to cost ratio]{
        \includegraphics[width=0.265\linewidth]{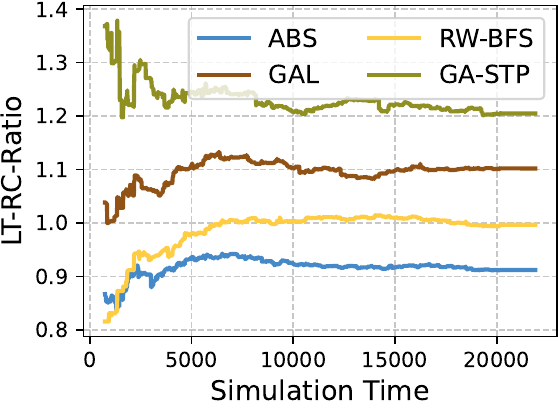}
        \label{fig: realltrc}
    }
    \vspace{-5pt}
    \caption{Acceptance ratio, long-term average revenue, and long-term revenue to cost ratio over simulation time for the random topology (upper) and the real-world topology (lower) according to Tab.~\ref{tab:results}. (All subsequent performance figures display only the best-performing algorithm from each category of solving methods.)}
    \label{fig: all1}
    \vspace{-3pt}
\end{figure*}
\vspace{-3pt}
\begin{figure}[t]
    \centering
    \subfigure[Random topology]{
        \includegraphics[width=0.46\linewidth]{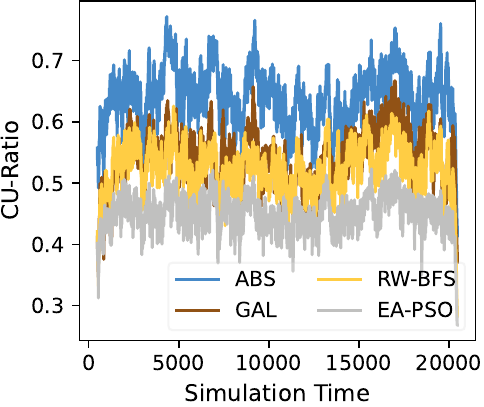}
        \label{fig: smallcur}
    }
    \hspace{-5pt}
    \subfigure[Real-world topology]{
        \includegraphics[width=0.46\linewidth]{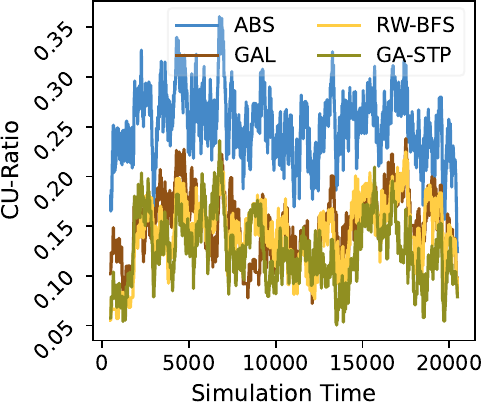}
        \label{fig: largecur}
    }
    \vspace{-5pt}
    \caption{Computing resource utilization comparison over
simulation time.}
    \label{fig: cur}
     \vspace{-5pt}
\end{figure}
RW-BFS maintains solid performance among heuristics by using topology-aware node ranking, often grouping nodes with similar scores and thus tending to co-locate \SF s with comparable roles in the network.
In contrast, RMD leverages graph partitioning to minimize the bandwidth requirement of \textit{Cut}-\LL s, ensuring co-location assignments are always optimal from the partitioning perspective.
However, its reliance on rigid local greedy strategies often leads to suboptimal co-location groups and mappings, ultimately resulting in the poorest overall performance.

Unexpectedly, meta-heuristic algorithms EA-PSO and GA-STP, despite their strong global search capabilities, deliver only moderate results.
Since both operate on independent node-level decisions and ignore the challenges introduced by \SF\ co-location, their best performance merely approaches that of RW-BFS, rather than surpassing it.

\subsubsection{Metric-based Analysis}
A comparison and discussion of each performance metric is provided here.

Acceptance ratio:
It directly impacts user satisfaction, as frequent service denials can erode user patience. 
As shown in Tab.~\ref{tab:results}, our \OURALGO\ achieves significantly higher acceptance ratios than other algorithms, with improvements of at least 15\% under the random topology and 60.2\% under the real-world topology.
However, the acceptance ratio alone has its limitations: it does not capture the quality of accepted requests. 
For example, given the same acceptance ratio, accepting a large number of requests with simple SEs may not lead to efficient resource utilization.

Revenue: This metric intuitively measures both the quantity and quality (complexity) of accepted requests, capturing the overall service gain.
Higher revenue also indicates more efficient resource utilization of the algorithm, as it creates more value from limited resources.
Therefore, combining revenue with acceptance ratio can provide a more comprehensive assessment of mapping quality.
As shown in Tab.~\ref{tab:results}, our \OURALGO \ achieves at least an 18\% and a 74.1\% improvement in total revenue under random and real-world topologies,  respectively.
Notably, compared to GAL, \OURALGO{} achieves a much greater improvement in total revenue than in acceptance ratio (e.g., 74.1\% vs. 60.2\% in the real-world topology).
This significant gap further indicates that \OURALGO{} is capable of accepting more complex and higher-value service requests, rather than merely increasing the number of accepted requests.


Long-term average revenue: 
It complements revenue by evaluating the ability to sustain high revenue throughout the simulation, reflecting long-term operational effectiveness rather than short-term bursts.
Comparing the trends of acceptance ratio and long-term average revenue provides insights into the decision-making quality of different algorithms. Specifically, our \OURALGO{} consistently exhibits a larger gap over other algorithms in long-term average revenue than in acceptance ratio, indicating a persistent preference for higher-quality requests.
As a counterexample, in Figs.~\ref{fig: randomacpt} and~\ref{fig: randomltr}, although GAL achieves a higher acceptance ratio than RW-BFS after simulation time 5000, its long-term average revenue does not show a corresponding increase. 
Normally, if the value distribution of requests remains unchanged, a higher acceptance ratio should result in a higher average revenue.
This suggests that GAL is accepting more small, less valuable requests during this period. A similar phenomenon can be observed in Figs.~\ref{fig: realacpt} and~\ref{fig: realltr}.
These findings further validate, from a long-term operational perspective, that \OURALGO{}’s mapping decisions favor complex, high-value service requests, rather than simply maximizing the number of accepted requests.



Profit:
Tab.~\ref{tab:results} presents the profit performance of each algorithm under the parameter settings ($\theta=2$, $\omega=0.5$), where \OURALGO\ consistently demonstrates a significant advantage.
It is essential to note that higher profits do not always require minimizing costs or maximizing the revenue-to-cost ratio. 
As long as costs are reasonable, prioritizing higher acceptance ratios and higher-quality requests can yield greater profit.



\begin{figure*}[t]
    \centering
    \subfigure[Random topology]{
        \includegraphics[width=0.85\linewidth]{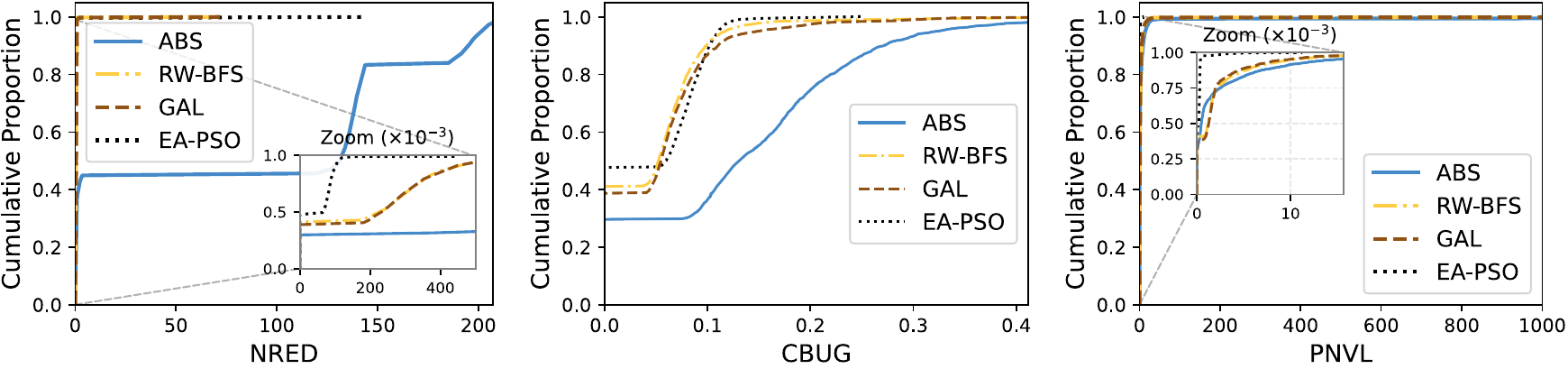}
        \label{motivation: average latency}
    }
    \\
    \subfigure[Real-world topology]{
        \includegraphics[width=0.85\linewidth]{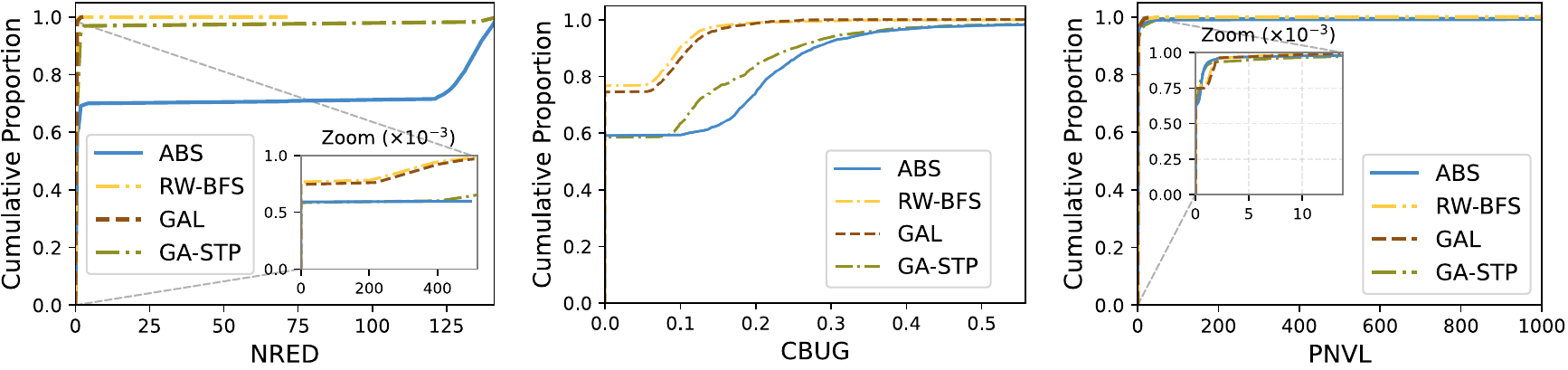}
        \label{motivation: average latency1}
    }
    \vspace{-5pt}
    \caption{CDF distributions of fragmentation evaluation metrics (NRED, CBUG, and PNVL) for the decisions on 2{,}000 test requests: \OURALGO\ versus each method's best-performing algorithm in the random topology (upper) and the real-world topology (lower). 
    For the combined fragmentation-based metric in (\ref{f29.5}), the weights $(\varpi_1,\varpi_2,\varpi_3)$ corresponding to (NRED, CBUG, PNVL) are selected via grid search and fixed in all experiments as $(0.52, 0.47, 0.20)$.}
    \vspace{-5pt}
    \label{fig: cdf}
\end{figure*}

\CN\ Resource Utilization: 
This metric directly reflects the effectiveness in integrating and utilizing distributed computing resources, which is the fundamental goal of CPN.
Once the acceptance ratio stabilizes, it provides an intuitive measure of an algorithm's resource exploitation capability. 
As illustrated in Fig.~\ref{fig: cur}, \OURALGO{} achieves improvements of at least 19.3\% and 73.2\% over baseline algorithms under the two topologies, respectively. 
Notably, the consistency between \CN\ resource utilization and revenue improvements mutually reinforces both metrics, highlighting that \OURALGO{}'s revenue advantage is fundamentally rooted in its superior exploitation of distributed computing resources. 
Accordingly, revenue becomes a meaningful indicator of CPN’s core objective---maximizing resource utilization---which justifies our preference for maximizing revenue, even when it may involve higher costs.

Revenue to Cost Ratio and Long-term Revenue to Cost Ratio:
These metrics primarily reflect differences in network resource consumption when mapping SEs.
Lower \NL\ bandwidth usage for the same SE implies lower cost and higher RC-Ratios.
However, in our work, maximizing computing resource utilization is the primary objective, with network bandwidth serving as the enabler rather than the end goal. 
Therefore, we are willing to tolerate higher network expenditure if it leads to greater acceptance ratios, revenues, and resource utilization; as such, a higher RC-Ratio is not always preferable.
As shown in Tab.~\ref{tab:results} and Figs.~\ref{fig: randomltrc} and \ref{fig: realltrc}, although \OURALGO \ is not optimal in terms of the RC-Ratio, its LT-RC-Ratio maintains a stable trend over time, without significant fluctuations.
This steady trend demonstrates that the trade-off between network consumption and improved service gains, as well as computing resource utilization, is both rational and sustainable, reflecting \OURALGO 's well-balanced and reliable resource utilization strategy.

\subsubsection{Extended Discussion}
\label{subsubsec:extends}
In the following, we further elaborate on two essential features of ABS. 

A key advantage of our approach is its effective global search capability, which becomes especially important in network resource-constrained scenarios--—a situation that may be commonly encountered in CPN environments.
Both GAL and ABS\textsubscript{init by RW-BFS} build upon the solutions generated by RW-BFS and further improve performance by leveraging global search capabilities.
For example, in the rocketfuel topology, where network resources are limited due to a higher number of \CN s but fewer \NL s (see Tab.~\ref{tab:parameter}), GAL’s global search increases revenue from 6.8\% (RW-BFS) to 10.1\%, a 48.5\% improvement, while our ABS\textsubscript{init by RW-BFS} raises it from 21.9\% to 67.6\%, corresponding to a 208.7\% relative improvement. 
These improvements are a direct result of compelling global exploration, and the much larger gain achieved by our method over GAL further demonstrates the superior global search capability of our approach.

Moreover, \OURALGO\ incorporates an evaluation mechanism more closely aligned with actual performance than traditional network cost-based metrics. 
Existing search-based algorithms (e.g., EA-PSO, GA-STP, RL-QoS, GAL) mainly rely on revenue-to-cost ratios or equivalent metrics as guidance for decision-making.
However, as shown in Tab.~\ref{tab:results}, optimizing RC-Ratio does not correlate with revenue gains, indicating that this kind of local efficiency metric cannot capture system complexity or predict actual performance in CPN. 
In contrast, our fragmentation-based metrics more accurately reflect global resource utilization, with higher values indicating lower fragmentation and higher efficiency. 
The CDF analysis in Fig.~\ref{fig: cdf} shows that \OURALGO\ predominantly achieves optimal metric concentrations across topologies, outperforming the top baselines and validating that our metrics reliably indicate and promote superior global performance.
Additionally, the degree of separation between the CDF curves of NRED, CBUG, and PNVL highlights their differing associations with overall performance. 
Specifically, NRED shows the strongest correlation, followed by CBUG, while PNVL has the least. This suggests a natural weighting hierarchy among these metrics when evaluating decision quality. 
Therefore, in practical applications, higher weights should be assigned to metrics more strongly correlated with overall performance.

\section{Conclusion}
\label{secs:conclusion}

In this paper, we have conducted a formal study of the service entity mapping problem in CPN environments, targeting optimal service-to-infrastructure assignment that fully leverages network-enabled computing resource integration. 
We have formally defined the problem, proved its theoretical intractability, and highlighted key practical challenges, including variable coupling, solution space explosion, and the need for global performance guidance.
We have proposed the ABS approach---a modular and adaptive search framework, instantiated in this work with a distributed particle swarm optimizer.
ABS introduces graph partitioning-based reformulation, bilevel optimization, and fragmentation-aware evaluation to address the identified challenges effectively.
Extensive simulations across diverse CPN scenarios demonstrate that ABS achieves significant and sustained improvements over state-of-the-art methods across key performance metrics and confirm the effectiveness of the fragmentation-based evaluation mechanism.
In future work, we will leverage the modular design of ABS to explore advanced search strategies (e.g., DRL) and incorporate additional CPN scenario features, such as resource multidimensionality and heterogeneity.

\bibliographystyle{IEEEtran}
\bibliography{reference}

\clearpage
\begin{appendices}
\section{Proofs of Theoretical Results}
\label{apdxA}
\setcounter{definition}{0}
\setcounter{lemma}{0}
\setcounter{proposition}{0}



\setcounter{theorem}{1}

\setcounter{corollary}{1}


\begin{theorem}
	\textup{($\mathbb{P}3$)} and \textup{($\mathbb{P}4$)} are equivalent.
	\label{t2}
\end{theorem}

\begin{proof}
	Let $\mathbf{x}$ be feasible for ($\mathbb{P}3$). Define, for all $m^s\in N^s$,
	\begin{align}
		\rho_{m^s} \triangleq 
		\frac{\sum_{u^v \in N^v_i} c(u^v)\, x^{u^v}_{m^s}}{\sum_{u^v \in N^v_i} c(u^v)}, 
		\quad
		z_{m^s} \triangleq \mathbb{I}\Big\{\sum_{u^v \in N^v_i} x^{u^v}_{m^s} \ge 1\Big\},
		\label{eq:rho_z_def}
	\end{align}
    where $\mathbb{I}\{\cdot\}$ is the indicator function.
	Then (\ref{f19f})--(\ref{f19e}) follow from (\ref{f1}).
	Moreover, (\ref{f19c}) holds with equality at the middle term by \eqref{eq:rho_z_def}.
	Finally, \eqref{eq:rho_z_def} implies (\ref{f19d}) and (\ref{f19d3}); and if $\rho_{m^s}>0$ then the numerator in \eqref{eq:rho_z_def} is positive, hence $z_{m^s}=1$ and (\ref{f19d2}) holds. Thus $(\boldsymbol{\rho},\mathbf{x},\mathbf{z})$ is feasible for ($\mathbb{P}4$).

	For any feasible $(\boldsymbol{\rho},\mathbf{x},\mathbf{z})$ of ($\mathbb{P}4$), $\mathbf{x}$ satisfies (\ref{f1})--(\ref{f3}) and is therefore feasible for ($\mathbb{P}3$).

	Hence the feasible sets of ($\mathbb{P}3$) and ($\mathbb{P}4$) coincide when projected onto $\mathbf{x}$. Since both problems minimize the same $\Phi(\mathbf{x})$, they have the same optimal objective value and the same optimal $\mathbf{x}^\star$.
\end{proof}

\setcounter{theorem}{2}
\begin{theorem}
	The PW-kGPP problem is NP-hard.
	\label{t3}
\end{theorem}

\begin{proof}
	The balanced graph partitioning (BGP) problem is a known problem~\cite{balanced}: given a graph $G = (V, E)$ and an integer $k \geq 2$, the task is to partition the vertex set $V$ into $k$ disjoint subsets $V_1, V_2, ..., V_k$ such that:
	the size of each subset satisfies $|V_p| \leq \lceil |V|/k \rceil for \, p = 1, 2, ..., k$ and the cut weight is minimized.
	It is very easy to map any BGP instance to a PW-kGPP instance, proceeding as follows:
	assign a weight of $w_n = 1$ to each vertex in $V$, define the proportion for each subset as $r_p = 1/k$ for all subsets $V_p$, and use a tolerance coefficient $\theta = 0$.
	Thus, any problem belonging to BGP can be solved by an algorithm designed for PW-kGPP; that is, \nolinebreak BGP can be reduced to PW-kGPP.
	Since BGP is NP-hard, PW-kGPP is also NP-hard.
\end{proof}

\begin{corollary}
	\textup{($\mathbb{P}4$)} is NP-hard.
    \label{coro2}
\end{corollary}

\begin{proof}
	According to the above analysis, a polynomial-time reduction exists that transforms any instance of PW-kGPP into an instance of ($\mathbb{P}4$).
	According to Theorem \ref{t2}, it can be proved that ($\mathbb{P}4$) is NP-hard.
\end{proof}

\begin{corollary}
	\textup{($\mathbb{P}3$)} is NP-hard.
	\label{coro3}
\end{corollary}

\begin{proof}
	Combining Theorem \ref{t2} and Corollary \ref{coro2}, ($\mathbb{P}3$) can be easily proved to be NP-hard.
\end{proof}

\begin{lemma}
Any global optimal solution $(\boldsymbol{\rho}^*,\mathbf{x}^*,\mathbf{z}^*)$ to \textup{($\mathbb{P}4$)}
is a global optimal solution to \textup{($\mathbb{P}4\textup{-BOP}$)}.
\label{l1}
\end{lemma}

\begin{proof}
Feasibility in ($\mathbb{P}4$) implies $\boldsymbol{\rho}^*$ satisfies (\ref{f19f})--(\ref{f19e}) 
and $(\mathbf{x}^*,\mathbf{z}^*)$ satisfies the lower-level constraints under $\boldsymbol{\rho}^*$.
If $(\mathbf{x}^*,\mathbf{z}^*)$ were not lower-level optimal, a feasible
$(\tilde{\mathbf{x}},\tilde{\mathbf{z}})$ under $\boldsymbol{\rho}^*$ with
$\Phi(\tilde{\mathbf{x}})<\Phi(\mathbf{x}^*)$ would yield a better feasible solution to ($\mathbb{P}4$),
contradicting optimality. Hence $(\mathbf{x}^*,\mathbf{z}^*)$ solves the lower-level and $\boldsymbol{\rho}^*$
is upper-level optimal.
\end{proof}

\begin{lemma}
Any global optimal solution $(\boldsymbol{\rho}^{\dagger},\mathbf{x}^{\dagger},\mathbf{z}^{\dagger})$ to \textup{($\mathbb{P}4\textup{-BOP}$)}
is a global optimal solution to \textup{($\mathbb{P}4$)}.
\label{l2}
\end{lemma}

\begin{proof}
By definition of ($\mathbb{P}4\textup{-BOP}$), $(\boldsymbol{\rho}^{\dagger},\mathbf{x}^{\dagger},\mathbf{z}^{\dagger})$
satisfies all constraints of ($\mathbb{P}4$), hence is feasible for ($\mathbb{P}4$).
If it were not globally optimal for ($\mathbb{P}4$), there would exist a feasible $(\bar{\boldsymbol{\rho}},\bar{\mathbf{x}},\bar{\mathbf{z}})$ to ($\mathbb{P}4$) such that
$\Phi(\bar{\mathbf{x}}) < \Phi(\mathbf{x}^{\dagger})$.
Then, for the upper-level decision $\bar{\boldsymbol{\rho}}$, the lower-level optimal value of ($\mathbb{P}4\textup{-BOP}$) is at most $\Phi(\bar{\mathbf{x}})$, implying an upper-level objective value strictly smaller than $\Phi(\mathbf{x}^{\dagger})$, which contradicts the global optimality of $(\boldsymbol{\rho}^{\dagger}, \mathbf{x}^{\dagger}, \mathbf{z}^{\dagger})$ for ($\mathbb{P}4\textup{-BOP}$).
Therefore, $(\boldsymbol{\rho}^{\dagger}, \mathbf{x}^{\dagger}, \mathbf{z}^{\dagger})$ is a global optimal solution to ($\mathbb{P}4$).
\end{proof}

\begin{theorem}
	Problems \textup{($\mathbb{P}4$)} and \textup{($\mathbb{P}4\textup{-BOP}$)} are equivalent in terms of global optimality: they have the same global optimal objective value and the same set of global optimal solutions.
    \label{t4}
\end{theorem}

\begin{proof}
    Following from Lemmas~\ref{l1} and \ref{l2}, any global optimal solution to ($\mathbb{P}4$) is also a global optimal solution to ($\mathbb{P}4\textup{-BOP}$), and vice versa. 
    Hence, the two problems share the same set of global optimal solutions. 
    Since both problems adopt the same objective function $\Phi(\cdot)$, they attain the same global optimal objective value.
\end{proof}

\begin{proposition}
	Suppose the outer layer can traverse all possible values of $\rho_{m^s}$, and the inner layer will solve the corresponding problems optimally. 
	The final solution obtained by the nested approach will be the optimal one for \textup{($\mathbb{P}2$)}.
	\label{T5}
\end{proposition}
\begin{proof}
	Suppose, for contradiction, that the nested approach does not yield the optimal solution to $(\mathbb{P}2)$. 
	Let $(\boldsymbol{\rho}^*, \mathbf{x}^*, \mathbf{f}^*)$ be an optimal solution. 
	Since the outer layer exhaustively traverses all possible $\boldsymbol{\rho}$, it will consider $\boldsymbol{\rho}^*$. 
	For this choice, the inner layer solves the corresponding PW-kGPP and IMCF problems optimally, yielding the minimum bandwidth resource cost for $\boldsymbol{\rho}^*$. 
	Therefore, the assumption is contradicted. 
\end{proof}
\section{Details of Definition 1}
\label{apdxc}

\begin{definition}[Proportional Weight-Constrained \textit{k}-way Graph Partitioning Problem, PW-kGPP]    
\begin{subequations}
	\begin{align}
	\label{f18a}
	\mathrm{PW}&\mathrm{-kGPP:} \min_{x^i_p} \sum_{(i, j) \in E} c_{ij} \sum_{p=1}^k x^i_p (1 - x^j_p) \\
	\label{f18b}
	&s.t. \quad \forall p \in \{1, 2, ..., k\}: \nonumber \\
	& \quad  \quad \, (1-\theta)r_p
	\leq \frac{\sum_{i \in V_p} w_i}{\sum_{i \in V} w_i}
	\leq (1+\theta)r_p,
	\end{align}
\end{subequations}
	where $\theta$ is a configurable tolerance coefficient.
	\label{D3}
\end{definition}
\section{Supplemental Materials for ABS Implementation}
\label{apdxb}

\begin{algorithm}[!h]
\SetAlgoLined
\scriptsize  
\KwIn{$G_i^v$, $G^s$, Worker Number $N_W$}
\KwOut{Best SEM Solution $(\mathbf{x^*,f^*})$}

Set archive capacity $N_A$, counter $c_{terminate} = 0$\;
Initialize the set of communication connections $\mathbf{conns}$ with $N_W$ $conn$s and archive $\mathbf{A} \leftarrow \emptyset$\;
Invoke Algorithm \ref{alg:worker} given $G_i^v$ and $G^s$ as input to activate $N_W$ workers, and assign each of them a $conn \in \mathbf{conns}$\;
\While{$c_{terminate} \neq N_W$}{
    \For{conn $\in \mathbf{conns}$}{
        \If{conn.has\_data()}{
            $(info, particle) \leftarrow$ conn.receive()\;
            \Switch{$info$}{
                \Case{`best\_particle'}{
                    \eIf{$|\mathbf{A}| < N_A$}{ 
                        Add $particle$ to $\mathbf{A}$\;
                    }{
                        Replace random $\mathbf{A}[r]$ if $\mathcal{F}(particle.solution) < \mathcal{F}(\mathbf{A}[r].solution)$\;
                    }
                }
                \Case{`newP\_request'}{
                    conn.send($\mathbf{A}[\text{rand}(|\mathbf{A}|)]$)\;
                }
                \Case{`terminate'}{
                    $c_{terminate} \leftarrow c_{terminate} + 1$\;
                }
            }
        }
    }
}
\Return $\argmin_{particle \in \mathbf{A}} \mathcal{F}(particle.solution)$\;
\caption{\OURALGO}
\label{alg:master}
\end{algorithm}

\begin{algorithm}[h]
\scriptsize  
\SetAlgoLined
\KwIn{$G_i^v$, $G^s$, Elite Set $\mathbf{ES}^{iter}$, Common Set $\mathbf{CS}$, Local Archive $\mathbf{LA}$}
\KwOut{Updated common set $\mathbf{CS}$}
\For{$particle \in \mathbf{CS}$}{
	\If{$iter == 1$}{
		Set $particle.dimension$ as the number of elements in $particle.position$\;
	}
	Set $r\_position = [0] * particle.position$.len()\;
	Get $temp\_position$ for $particle$ using (\ref{f30})-(\ref{f33})\;
	$particle.position = [max(0, \rho) \, for \, \rho \,\, in \, temp\_position]$\;
	Find the $particle.dimension$ largest elements of $particle.position$, normalize them, and set them into their corresponding positions in $r\_position$\;
	$\mathbf{x'} \leftarrow \, metispartition(\cdot)$, $\mathbf{f'} \leftarrow \, kshortestpath(\cdot)$\;
	\If{($\mathbf{x'}, \mathbf{f'})$ satisfies (\ref{f1})-(\ref{f6})}{
		$particle.dimension$ = \\ max(1, $particle.dimension-1$)
	}
}
\caption{Swarm Evolutionary Update}
\label{alg:evolve}
\end{algorithm}

\begin{algorithm}[h]
\SetAlgoLined
\scriptsize  
\KwIn{$G_i^v$, $G^s$, Communication Connection $conn$}
Set swarm size $N_S$, max iteration times $ITR_{max}$, elite set capacity $N_E$, local archive capacity $N_{LA}$, particle with global best fitness $gbest=\{\}$\;
Initialize the beginning elite set $\mathbf{ES}^0 \leftarrow \emptyset$, common set $\mathbf{CS} \leftarrow \emptyset$, and local archive $\mathbf{LA} \leftarrow \emptyset$\;

Invoke Algorithm \ref{alg:init} given $G_i^v$, $G^s$, and $N_S$ as input to get initial $\mathbf{Particles}$\;
\For{$iter = 1$ \KwTo $ITR_{max}$}{
	Compute the fitness value $\mathcal{F}(particle.solution)$ for $particle \in \mathbf{Particles}$\;
	Sort $\mathbf{Particles}$ by fitness values, update $gbest$, and put $N_E$ best particles into $\mathbf{ES}^{iter}$ and $(N_S - N_E)$ worst $particle$s into $\mathbf{CS}$\;
	\If{$gbest$ is updated}{
		$conn$.send(\textit{`best\_particle'}, $gbest$)\;
	}
	\If{$\mathbf{ES}^{iter} == \mathbf{ES}^{iter-1}$}{
		$conn$.send(\textit{`newP\_request'})\;
		$n\_particle$ = $conn$.receive()\;
		\eIf{$\mathbf{LA}$.len() $<$ $N_{LA}$}{
			Add $n\_particle$ to $\mathbf{LA}$\;
		}{
			Replace the worst particle $w\_particle \in \mathbf{LA}$ with $n\_particle$, if $\mathcal{F}(n\_particle.solution)$ is smaller\;
		}
	}
	Invoke Algorithm \ref{alg:evolve} given $\mathbf{CS}$, $\mathbf{ES}^{iter}$, and $\mathbf{LA}$ as input to 
	make $particle \in \mathbf{CS}$ evolve\;
	$iter = iter + 1$\;
}
$conn$.send(\textit{`terminate'})\; 
\caption{\OURALGO \ Iteration Worker}
\label{alg:worker}
\end{algorithm}

\begin{algorithm}[h]
\SetAlgoLined
\scriptsize  
\KwIn{$G_i^v$, $G^s$, Swarm Size $N_S$}
\KwOut{Particle Swarm $\mathbf{Particles}$}
Initialize $\mathbf{Particles} \leftarrow \emptyset$ and $\boldsymbol{\rho}, \mathbf{x}, \mathbf{f} \leftarrow None$\;
\For{$i = 1$ \KwTo $N_S$}{
	$(\boldsymbol{\rho}, \mathbf{x}) \leftarrow$ \textit{init\_solver($G_i^v$, $G^s$)}\;
	$\mathbf{f} \leftarrow$ \textit{kshortestpath($G_i^v$, $G^s$, $\mathbf{x}$)}\;
	Add particle$(\boldsymbol{\rho}, (\mathbf{x},\mathbf{f}))$ to $\mathbf{Particles}$\;
}
\Return $\mathbf{Particles}$\;

c
\SetKwFunction{FInit}{init\_solver}
\SetKwProg{Fn}{Function}{:}{}
\Fn{\FInit{$G_i^v$, $G^s$}}{
	Set $max\_depth, depth = 0$; $\boldsymbol{\rho}, \mathbf{x} \leftarrow None$\;
	Initialize the chosen \CN \ set, current neighbor set, next neighbor set, and useless neighbor set: $\mathbf{chosen}, \mathbf{cNBR}, \mathbf{nNBR}, \mathbf{uNBR} \leftarrow \emptyset$\;
	Select $m^s$ with $C(m^s)>0$ by resource-weighted randomization; update $\mathbf{chosen}$ and $\mathbf{cNBR}$\;
	\While{$\mathbf{chosen}$.size() $< min(|N^s|, |N_i^v|)$ \textbf{and} $depth \leq max\_depth$}{
		\eIf{$\mathbf{cNBR} \neq \emptyset$}{
			Select $m^s$ from $\mathbf{cNBR}$ via resource-weighted randomization, add it to $\mathbf{chosen}$, remove it from $\mathbf{cNBR}$\;
			For each neighbor of $m^s$, if it is not in $\mathbf{chosen}$, add it to $\mathbf{nNBR}$ if it has non-zero resources; otherwise, add it to $\mathbf{uNBR}$\;
			Calculate $\rho_{m^s}$ as: $\forall m^s \in \mathbf{chosen}, \rho_{m^s} = C(m^s) / \sum_{m^s \in \mathbf{chosen}} C(m^s_i)$, set $\rho_{m^s} = 0$ for $m^s \notin \mathbf{chosen}$, and update $\boldsymbol{\rho}$\;
			$\mathbf{x'} \leftarrow$ \textit{metispartition}($G^v_i$, $\boldsymbol{\rho}$)\;
			\If{$\mathbf{x'}$ satisfies constraints (\ref{f1})-(\ref{f3})}{set $\mathbf{x} = \mathbf{x'}$; \textbf{break}}
		}{
			\If{$\mathbf{nNBR} \neq \emptyset$}{
				Set $\mathbf{cNBR}$ as $\mathbf{nNBR}$, $\mathbf{nNBR} \leftarrow \emptyset$, and $\mathbf{uNBR} \leftarrow \emptyset$\;
			}
			\uElseIf{$\mathbf{uNBR} \neq \emptyset$}{
				Copy $\mathbf{uNBR}$ to $\mathbf{uNBR'}$ and set $\mathbf{nNBR}, \mathbf{uNBR} \leftarrow \emptyset$\;
				For each neighbor of $m^s \in \mathbf{uNBR'}$, if it is not in $\mathbf{chosen}$, add it to $\mathbf{nNBR}$ if it has non-zero resources; otherwise, add it to $\mathbf{uNBR}$\;
				$max\_depth = max\_depth + 1$\;
			}
			\Else{break\;}
		}
		
	}
	\Return $(\boldsymbol{\rho}, \mathbf{x})$\;
}
\caption{Particle Swarm Initializing}
\label{alg:init}
\end{algorithm}

\subsection{Pseudocode for Algorithms 1 to 4}
The pseudocode can be found in Algorithms~\ref{alg:master}, \ref{alg:evolve}, \ref{alg:worker}, \ref{alg:init}.

\subsection{Details of Time Complexity Analysis}
In practical scenarios, the most time-consuming component of the algorithm is the iterative process executed by the workers (Algorithm \ref{alg:worker}), which forms the basis of our complexity analysis.
During initialization (Algorithm~\ref{alg:init}), \nolinebreak the \textit{init\_solver} and \textit{kshortestpath} functions are invoked $|N_S|$ times.
The dominant complexity in \textit{init\_solver} arises from the \textit{metispartition} function.
Consequently, the overall complexity of this phase is: $\mathcal{O}(|N_S| \cdot (|L^s|+k\cdot|N^s|\cdot(|L^s|+|N^s|\log|N^s|)))$. \nolinebreak
In the iterative main body, the algorithm performs $ITR_{max}$ iterations. \nolinebreak
For each iteration, the swarm evolutionary update (Algorithm \ref{alg:evolve}) processes $|N_S|$ particles. 
The most time-consuming operations for each particle are those invocations of the \textit{metispartition} and \textit{kshortestpath} functions. \nolinebreak
Therefore, the complexity of this phase is: $\mathcal{O}(ITR_{max}\cdot|N_S| \cdot (|L^s|+k\cdot|N^s|\cdot(|L^s|+|N^s|\log|N^s|)))$.
By combining the complexities of the two phases, the overall complexity of the algorithm executed by each worker can be simplified as: $\mathcal{O}(ITR_{max}\cdot|N_S| \cdot k\cdot|N^s|\cdot(|L^s|+|N^s|\log|N^s|))$.

\section{Experimental Details}
\label{apdxd}

\subsection{Compared Methods}

\textit{Heuristic}--
\begin{itemize}
	\item \textbf{RW-BFS}~\cite{rw} is a classic rank-based heuristic considering CPN topology and ranking each \CN \ via random walk.
		The \SF \ mapping and correlated \LL \ mapping are conducted simultaneously based on the node rank by breadth-first searching (BFS).
	\item \textbf{RMD}~\cite{rmd} introduces a service entity topology preprocessing technology and uses a heuristic strategy for mapping, which achieves a topology-aware solution through graph coarsening, partitioning, uncoarsening, and merging nodes with strong link relations.
\end{itemize}

\textit{Meta-heuristic}--
\begin{itemize}
	\item \textbf{EA-PSO}~\cite{psovne} is a classic meta-heuristic that adopts a discrete PSO algorithm for directly searching the \SF \ mapping solution, along with the shortest path algorithm to map the related \LL s. 
	\item \textbf{GA-STP}~\cite{gastp} relies on Genetic Algorithm (GA) to jointly coordinate \SF \ and \LL \ mappings. Meanwhile, a novel heuristic conciliation mechanism is embedded in the exploration process of the GA algorithm to handle a set of potentially infeasible \LL \ mapping solutions.
\end{itemize}

\textit{Learning-based}--
\begin{itemize}
	\item \textbf{RL-QoS}~\cite{pgcnn} is a model-free reinforcement learning algorithm that leverages a policy gradient method with a convolutional neural network (CNN) and softmax layers to make node decisions based on historical data. It estimates the policy network through the CNN and updates the gradient using the policy gradient algorithm.
	\item \textbf{GAL}~\cite{gal} employs a two-stage framework: initially, a graph neural network is trained with supervised learning to emulate classic node ranking (e.g., RW-BFS), followed by reinforcement learning for further improvement, with final \SF\ and \LL\ mapping performed using a BFS strategy. 
\end{itemize}

\subsection{Evaluation Metrics}

Denote $\Upsilon(t)$ as the set of all arrived SE topologies and $\Upsilon^a(t)$ as the subset of successfully mapped ones.
We set $w_c=w_b=\pi_c=\pi_b=1$ in (\ref{f7}) and (\ref{f8}) following~\cite{gal,rw}.
\begin{itemize}
    \item \textbf{Acceptance ratio:} This metric is formulated in (\ref{f40}), and it directly reflects the proportion of service entities successfully mapped by the algorithm.
\begin{equation}
	p_{ac}(t) = \frac{|\Upsilon^a(t)|}{|\Upsilon(t)|}.
	\label{f40}
\end{equation}

    \item \textbf{Revenue:} This metric reflects the total revenue generated by all successfully accepted service requests at a given time $t$.
    (\ref{f34}) shows the formula.
    \begin{align}
	\sum_{G^v_i \in \Upsilon^a(t)} \mathcal{R}(G_i^v).
	\label{f34}
    \end{align}

    \item \textbf{Long-term average revenue (LT-AR):} It complements the total revenue, measuring the ability of the algorithm to generate revenue throughout the simulation period continuously\footnote{Rather than infinite time horizons, we approximate these long-term metrics by cumulative metrics at the simulation time of the last arrival request.}.
    This metric is expressed in (\ref{f35}):
    \begin{align}
	\lim_{T \rightarrow \infty}\frac{\sum^T_{t=0} \mathcal{R}(G^v(t))}{T}.
	\label{f35}
    \end{align}

    \item \textbf{Profit:} It is formulated in (\ref{f41}). This metric integrates revenue, acceptance ratio, and cost to comprehensively evaluate the practical value of the algorithm, reflecting both the efficient utilization of computing resources and the overall service satisfaction.
\begin{align}
	(p_{ac}(t))^\varkappa \cdot \big(\sum_{G^v_i \in \Upsilon^a(t)} \mathcal{R}(G_i^v) - \omega \cdot \sum_{G^v_i \in \Upsilon^a(t)} \mathcal{C}(G_i^v)\big),
	\label{f41}
\end{align}
where $\varkappa \geq 1$ is a configurable weighting coefficient that adjusts the relative importance of the service acceptance ratio for the system's achieved profit;
and $0< \omega < 1$ serves as a tunable weighting factor that determines the system’s willingness to trade additional network resource consumption for higher computing resource utilization.

    \item \textbf{\CN\ Resource Utilization (CU-Ratio):} This metric indicates the overall utilization of \CN s within the CPN topology by service workloads at a given time $t$. 
    It represents the ratio of computing resources utilized by all currently accepted service requests at a specific time to the total initial computing capacity available across all \CN s, which is expressed in (\ref{f113}):
    \begin{align}
    \frac{\sum_{m^s \in N^s} U_C(m^s, t)}{
    \sum_{m^s \in N^s} C_{init}(m^s)},
	\label{f113}
    \end{align}
    where $U_C(m^s, t)$ is the instantaneous computing resource usage of CN $m^s$ at a timestamp $t$, and $C_{init}(m^s)$ denotes the initial computing capacity of the corresponding node.   

    \item \textbf{Revenue to cost ratio (RC-Ratio):} This metric quantifies the efficiency of mapping decisions by measuring the ratio between the cumulative revenue and cost of all accepted requests at time $t$, as shown in (\ref{f36}).
    \begin{align}
	\frac{\sum_{G^v_i \in \Upsilon^a(t)} \mathcal{R}(G_i^v)}{\sum_{G^v_i \in \Upsilon^a(t)} \mathcal{C}(G_i^v)}.
	\label{f36}
    \end{align}
    We do not use cost alone as a performance metric because, by definition, although it reflects the overhead of LLnM, higher cost often indicates that more services have been mapped. Therefore, cost alone does not meaningfully reflect algorithm efficiency. Instead, the RC-ratio better captures bandwidth usage efficiency by relating total revenue to total cost.

    \item \textbf{Long-term revenue to cost ratio (LT-RC-Ratio):} It calculates the ratio of long-term average revenue and cost generated during the service provisioning process.
    The formula is shown in (\ref{f37}).
    \begin{align}
	\lim_{T \rightarrow \infty} \frac{\sum^T_{t=0} \mathcal{R}(G^v(t))}{\sum^T_{t=0} \mathcal{C}(G^v(t))}.
	\label{f37}
    \end{align}
\end{itemize}
\end{appendices}

\vfill

\end{document}